\documentclass[a4paper,10pt]{article}

\usepackage{amsmath,amsgen,latexsym}
\usepackage{amstext,amssymb,amsfonts,latexsym}
\usepackage{theorem}
\usepackage{pifont}
\usepackage[hyphens]{url}
\usepackage{microtype}
\usepackage{graphicx}

 \newcommand{\bs}{\bigskip}
 
 \newcommand{\n}{\noindent}
 \newcommand{\s}{\smallskip}
 \newcommand{\hs}[1]{\hspace*{ #1 mm}}
 \newcommand{\vs}[1]{\vspace*{ #1 mm}}

 \newcommand{\nat}{\mathbb{N}}
 
 \newcommand{\integer}{\mathbb{Z}}

 \newcommand{\prob}{{\mathrm{Prob}}}

 \newcommand{\co}{\mathrm{co}\mbox{-}}

 \newcommand{\CC}{{\cal C}}
 \newcommand{\FF}{{\cal F}}

 \newcommand{\GG}{{\cal G}}

 \newcommand{\reg}{\mathrm{REG}}
 \newcommand{\cfl}{\mathrm{CFL}}
 \newcommand{\ucfl}{\mathrm{UCFL}}
 \newcommand{\dcfl}{\mathrm{DCFL}}

\theoremstyle{plain}
\theoremheaderfont{\bfseries}
 \newtheorem{theorem}{Theorem}[section]
 \newtheorem{lemma}[theorem]{Lemma}
 \newtheorem{proposition}[theorem]{{\bf Proposition}}
 \newtheorem{corollary}[theorem]{Corollary}

 {\theorembodyfont{\rmfamily}
  \newtheorem{definition}[theorem]{Definition}}
 {\theorembodyfont{\rmfamily} }
 {\theorembodyfont{\rmfamily} }

 \newtheorem{claim}{Claim}

 \newenvironment{proof}{\par \noindent
            {\bf Proof. \hs{2}}}{\hfill$\Box$ \vspace*{3mm}}

 \newenvironment{yproof}{\par \noindent
            {\bf Proof. \hs{2}}}{\hfill$\Box$ \vspace*{3mm}}

 \newenvironment{proofsketchof}[1]{\vspace*{5mm} \par \noindent
            {\bf Proof Sketch of #1.\hs{2}}}{\hfill$\Box$ \vspace*{3mm}}

 \newenvironment{proofof}[1]{\vspace*{5mm} \par \noindent
         {\bf Proof of #1.\hs{2}}}{\hfill$\Box$ \vspace*{3mm}}

\newcommand{\ignore}[1]{}

\newcommand{\track}[2]{[\: \begin{subarray}{c} #1 \\%
      #2 \end{subarray} ]}
\newcommand{\cent}{{|}\!\!\mathrm{c}}
\newcommand{\dollar}{\$}

\newcommand{\bpcfl}{\mathrm{BPCFL}}
\newcommand{\pcfl}{\mathrm{PCFL}}

\newcommand{\twopfa}{2\mathrm{PFA}}
\newcommand{\tworfa}{2\mathrm{RFA}}
\newcommand{\oneppda}{1\mathrm{PPDA}}
\newcommand{\onerpda}{1\mathrm{RPDA}}

\newcommand{\rcfl}{\mathrm{RCFL}}

\newcommand{\rtnfamv}{\mathrm{rtNFAMV}}

\newcommand{\onelpa}{1\mbox{-}\mathrm{LPA}}
\newcommand{\twolpa}{2\mbox{-}\mathrm{LPA}}

\newcommand{\twolda}{2\mbox{-}\mathrm{LDA}}

\newcommand{\onelbpa}{1\mbox{-}\mathrm{LBPA}}
\newcommand{\twolbpa}{2\mbox{-}\mathrm{LBPA}}
\newcommand{\onelra}{1\mbox{-}\mathrm{LRA}}
\newcommand{\twolra}{2\mbox{-}\mathrm{LRA}}
\newcommand{\klda}[1]{#1 \mbox{-}\mathrm{LDA}}
\newcommand{\klna}[1]{#1 \mbox{-}\mathrm{LNA}}
\newcommand{\klpa}[1]{#1 \mbox{-}\mathrm{LPA}}

\newcommand{\klra}[1]{#1 \mbox{-}\mathrm{LRA}}
\newcommand{\klbpa}[1]{#1 \mbox{-}\mathrm{LBPA}}
\newcommand{\klua}[1]{#1 \mbox{-}\mathrm{LUA}}
\newcommand{\mmid}{\!\mid\!}

\begin{document}

\pagestyle{plain}
\pagenumbering{arabic}
\setcounter{page}{1}

\begin{center}
{\Large {\bf Behavioral Strengths and Weaknesses of Various Models of Limited Automata}}\footnote{This current paper completes and corrects its  preliminary version, which appeared in the Proceedings of the 45th International Conference on Current Trends in Theory and Practice of Computer Science (SOFSEM 2019), Nov\'{y} Smokovec, Slovakia, January 27--30, 2019, Lecture Notes in Computer Science, Springer, vol. 11376, pp. 519--530, 2019.} \bs\\
{\sc Tomoyuki Yamakami}\footnote{Current Affiliation: Faculty of Engineering, University of Fukui, 3-9-1 Bunkyo, Fukui 910-8507,  Japan} \bs\\
\end{center}


\begin{abstract}
We examine the behaviors of various models of $k$-limited  automata, which naturally extend Hibbard's [Inf. Control, vol. 11, pp. 196--238, 1967] scan limited automata, each of which is a single-tape linear-bounded automaton satisfying the $k$-limitedness requirement that the content of each tape cell should be modified only during the first $k$ visits of a tape head.
One central computation model is a probabilistic $k$-limited automaton (abbreviated as a $k$-lpa), which  accepts an input exactly when its accepting states are reachable from its  initial state with probability more than 1/2 within expected polynomial time. We also  study the behaviors of one-sided-error and bounded-error variants of such $k$-lpa's as well as the deterministic, nondeterministic, and unambiguous models of $k$-limited automata, which can be viewed as natural restrictions of $k$-lpa's.
We discuss fundamental properties of these machine models and obtain inclusions and separations among language families induced by them.
In due course, we study special features---the blank skipping property and the closure under reversal---which are keys to the robustness of $k$-lpa's.

\s
\n{\bf Keywords.}
{limited  automata, pushdown automata, probabilistic computation, bounded-error probability, one-sided error, blank skipping property, reversal}
\end{abstract}

\sloppy

\section{Historical Background, Motivations, and Chellenges}\label{sec:introduction}

\subsection{Limited  Automata}\label{sec:ppda-intro}

Over the past 6 decades, automata theory has made a remarkable progress to unearth hidden structures and properties of various types of finite-state-controlled machines, including fundamental computation models of finite(-state) automata and one-way pushdown automata.

In an early period of the development of automata theory, Hibbard \cite{Hib67} introduced a then-novel rewriting system of so-called \emph{scan limited automata} in hope of characterizing context-free and deterministic context-free languages by direct simulations of their underlying one-way pushdown automata. Unfortunately, Hibbard's model seems to have been paid little attention until Pighizzini and Pisoni \cite{PP14,PP15} reformulated the model from a modern-machinery perspective and reproved a characterization theorem of Hibbard in a more sophisticated manner.
A \emph{$k$-limited automaton},\footnote{Hibbard's original formulation of ``$k$-limited automaton'' is equipped with a semi-infinite tape that stretches  only to the right with no endmarker but is filled with the blank symbols outside of an input string. Our definition in this paper is different from Hibbard's but it is rather similar to Pighizzini and Pisoni's \cite{PP14,PP15}.} for each fixed index $k\geq0$, is roughly  a one-tape (or a single-tape) Turing machine\footnote{A single-tape Turing machine model was discussed in the past literature, including \cite{Hen65,TYL10}.}
whose tape head is allowed to rewrite each tape cell between two endmarkers only during the first $k$ scans or visits (except that, whenever a tape head makes a ``turn,'' we count this move as ``double'' visits).
after the $k$ visits to a tape cell, the last symbol in the tape cell becomes unrewritable and frozen forever.
Although these automata can be viewed as a special case of linear-bounded finite automata, the restriction on the number of times that they rewrite tape symbols brings in quite distinctive effects on the computational power of the underlying automata, different from other restrictions, such as upper bounds on the  numbers of nondeterministic choices or the number of tape-head turns.
Hibbard conducted an intensive study on deterministic and nondeterministic behaviors of $k$-limited automata. In his study, he discovered that \emph{nondeterministic $k$-limited automata} (abbreviated as $k$-lna's) for $k\geq2$ are exactly as powerful as 1npda's, whereas 1-lna's are equivalent in power to 2-way deterministic finite automata (or 2dfa's) \cite{WW86}.
This gives natural characterizations of 1npda's and 2dfa's in terms of access-controlled Turing machines.

Another close relationship was proven by Pighizzini and Pisoni \cite{PP15} between \emph{deterministic $k$-limited automata} (or $k$-lda's) and one-way deterministic pushdown automata (or 1dpda's).
In fact, they proved that $2$-lda's embody exactly the power of 1dpda's. This equivalence in computational complexity contrasts  Hibbard's  observation that, for each index $k\geq2$, $(k+1)$-lda's in general possess more recognition power than $k$-lda's.
These phenomena exhibit a clear structural difference between determinism and nondeterminism on the machine model of ``limited  automata'' and such a difference naturally raises an important question of whether other variants of limited  automata can match their corresponding models of one-way pushdown automata in computational power.

\subsection{Extension to Probabilistic and Unambiguous Computations}

Lately, a computation model of \emph{one-way probabilistic pushdown automata} (or 1ppda's) has been discussed extensively to demonstrate computational strengths as well as weaknesses in \cite{HS10,KGF97,MO98,Yam17}.

While nondeterministic computation is purely a theoretical notion, probabilistic computation could be implemented in real life by installing a mechanism of generating (or sampling) random bits (e.g., by  flipping fair or biased coins).
From a generic viewpoint, deterministic and nondeterministic computations can be seen merely as restricted variants of probabilistic computation by appropriately defining the criteria for ``error probability'' of computation.
A \emph{bounded-error} probabilistic machine makes error probability bounded away from $1/2$, whereas an \emph{unbounded-error} probabilistic machine allows error to take arbitrarily close to probability $1/2$.

In many cases, a probabilistic approach helps us solve a target mathematical problem  algorithmically faster, and probabilistic (or randomized) computation often exhibits its superiority over its deterministic counterpart even on simple machine models.
For example, as Freivalds \cite{Fre81b} demonstrated,
2-way probabilistic finite automata (or 2pfa's) running in expected exponential time can recognize non-regular languages with bounded-error probability \cite{Fre81b}. By contrast, when restricted to expected subexponential runtime, bounded-error 2pfa's recognize only regular languages \cite{DS92,KF90}.
As this example shows, the expected runtime bounds of probabilistic machines largely affect the computational power of the machines, and thus its probabilistic behaviors   significantly differ from deterministic behaviors. In many studies, the runtime of machines is limited to expected polynomial time.
Probabilistic variants of pushdown automata were discussed intensively by Hromkovi\v{c} and Schnitger \cite{HS10} as well as Yamakami \cite{Yam17}. They demonstrated clear differences in computational power between two pushdown models, 1npda's and 1ppda's.

The aforementioned usefulness of probabilistic algorithms motivates us to take a probabilistic approach toward Hibbard's original model of $k$-limited automata. When $k$-limited automata are allowed to err, these machines are naturally expected to exhibit significantly better performance in computation. This paper in fact aims at introducing a novel model of \emph{probabilistic $k$-limited automata} (or $k$-lpa's) and their natural variants, including one-sided-error, bounded-error, and unbounded-error models restricted to expected polynomial running time,  and to explore their fundamental properties to obtain strengths and
weaknesses of families of languages
recognized by those machine models.
Since $k$-lda's and $k$-lna's are viewed as special cases of $k$-lpa's, many properties of $k$-lda's and $k$-lna's can be discussed within a wider framework of $k$-lpa's.

\emph{Unambiguity} has been paid special attention in formal languages and automata theory. The unambiguity for finite automata infers a unique accepting computation path if any. Stearns and Hunt III \cite{SH85}, for instance, demonstrated that, for converting a nondeterministic finite automaton into an equivalent finite machine, an unambiguous finite automaton  performs no better than any deterministic finite automaton.
For unambiguous context-free languages, there are known efficient algorithms to parse given words of the languages.
Herzog \cite{Her97} further generalized this notion for pushdown automata and discussed the amount of ambiguity. With the use of polynomial-size Karp-Lipton advice, as Reinhardt and Allender \cite{RA00} demonstrated,  for logarithmic-space auxiliary pushdown automata, it is possible to  make nondeterministic computation unambiguous.
In this paper, we also discuss \emph{unambiguous $k$-limited automata} (abbreviated as $k$-lua) as an unambiguous model of $k$-lna's. We wish to ask what relationships are met between pushdown automata and $2$-limited automata of the same machine types.

Let us introduce the notation for language families induced by the aforementioned models of $k$-limited automata. The notation $\klpa{k}$ denotes the family of all languages recognized by expected-polynomial-time $k$-lpa's with unbounded-error probability. In a way similar to $\klpa{k}$, one-sided-error and bounded-error models of $k$-lpa's induce $\klra{k}$ and $\klbpa{k}$, respectively. Furthermore, $\klda{k}$, $\klna{k}$, and $\klua{k}$ are obtained from $\klpa{k}$ respectively by replacing $k$-lpa's with $k$-lda's, $k$-lna's, and $k$-lua's. Containment relations among those language families are summarized in Fig.~\ref{fig:class-hierarchy}.


\begin{figure}[t]
\centering
\includegraphics*[height=5.0cm]{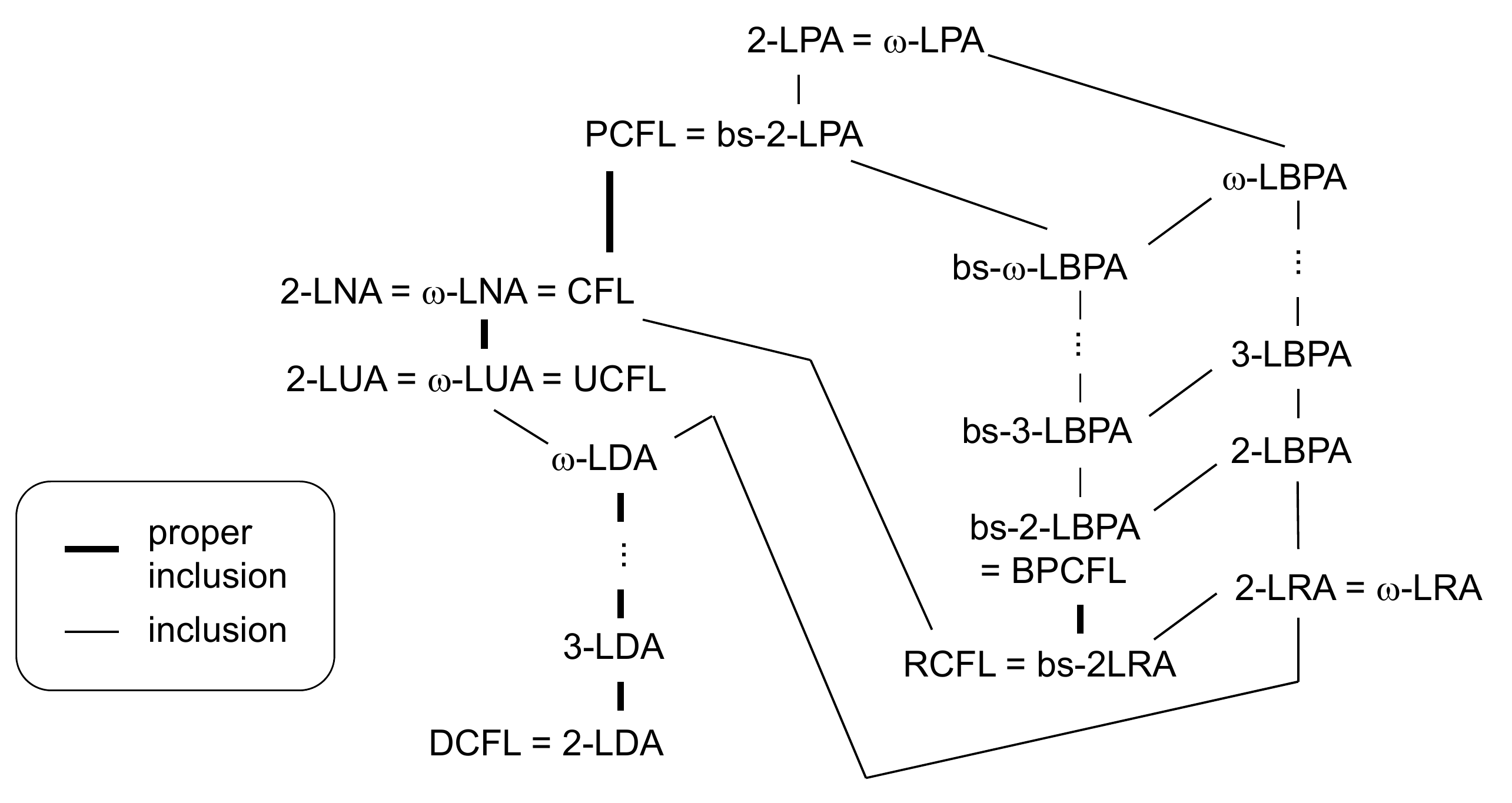}
\caption{Containment relations among language families discussed in this paper. Each upper family contains its lower family except for $\klda{\omega}\subseteq \klra{\omega}$.}\label{fig:class-hierarchy}
\end{figure}


\paragraph{Organization of This Work}\label{sec:major-contribution}

In Section \ref{sec:formal-definition}, we will give a formal definition of $k$-lpa's, following the existing models of 1ppda's explained in detail through Section \ref{sec:machine-model}.
Section \ref{sec:ideal-shape-lemma} argues how to convert any 1ppda to its pop-controlled form, called ``ideal-shape''.
We will present a basic property of ``blank skipping'' of $k$-lpa's, which is useful together with ``ideal shape'' in proving Theorem \ref{LPFA=PPDA} in Section \ref{sec:application-BSP}.
The collapses of language families induced by $k$-lpa's and by $k$-lua's as in Theorem \ref{LPFA-and-PCFL} will be proven in Section \ref{sec:property-omega-LPFA}.
We will discuss closure/non-closure properties of $k$-lda's in Section \ref{sec:closure-deterministic} and $k$-lpa's in Section \ref{sec:closure-probabilistic}.
Numerous problems left unanswered throughout this paper will be listed in Section \ref{sec:discussion} as a guiding compass to the future study of $k$-limited automata.

\section{Fundamental Notions and Notation}\label{sec:limited-finite-automata}

Let us formally introduce various computational models of limited  automata, in which we can
rewrite the content of each tape cell only during the first $k$ scans or visits of the cell. In comparison, we also describe probabilistic pushdown automata.

\subsection{Numbers, Alphabets, and Languages}

Let $\nat$ denote the set of all \emph{natural numbers}, which are  nonnegative integers, and set $\nat^{+}$ to be $\nat-\{0\}$. We denote by $[m,n]_{\integer}$ the set $\{m,m+1,m+2,\ldots,n\}$ for any two integers $m$ and $n$ with $m\leq n$.
This set $[m,n]_{\integer}$ is conveniently called an \emph{integer interval} in comparison to a real interval. In addition, we abbreviate as $[m]$ the integer interval $[1,m]_{\integer}$ for any integer $m\geq1$.

A nonempty finite set of ``symbols'' or ``letters'' is called an \emph{alphabet}. A \emph{string} $x$ over alphabet $\Sigma$ is a finite sequence of symbols taken from $\Sigma$ and its \emph{length} $|x|$ expresses the total number of symbols in $x$. The \emph{empty string} is a unique string of length $0$ and is always denoted $\lambda$.
The notation $\Sigma^*$ denotes the set of all strings over $\Sigma$ and any subset of $\Sigma^*$ is called a \emph{language} over $\Sigma$. Given a number $n\in\nat$, $\Sigma^n$ (resp., $\Sigma^{\leq n}$) expresses the set of all strings of length exactly $n$ (resp., at most $n$) over $\Sigma$. Obviously, $\Sigma^*$ coincides with $\bigcup_{n\in\nat}\Sigma^n$.

Given two alphabets $\Sigma_1$ and $\Sigma_2$, we construct a new alphabet $\{\track{\sigma}{\tau}\mid \sigma\in \Sigma_1,\tau\in\Sigma_2\}$ using the \emph{track notation} of \cite{TYL10}. A string over this alphabet is of the form  $\track{\sigma_1}{\tau_1}\track{\sigma_2}{\tau_2}\cdots \track{\sigma_n}{\tau_n}$, which is further abbreviated as $\track{x}{z}$ for $x=\sigma_1\sigma_2\cdots \sigma_n$ and $z=\tau_1\tau_2\cdots \tau_n$. To provide such a string $\track{x}{z}$ to an input tape, we split the tape into two tracks so that the upper track holds $x$ and the lower track does $z$.

Given a string $x$ of the form $\sigma_1\sigma_2\cdots \sigma_{n-1}\sigma_n$ over alphabet $\Sigma$ with $\sigma_i\in\Sigma$ for all $i\in[n]$, the \emph{reverse} of $x$ is $\sigma_n\sigma_{n-1}\cdots \sigma_2\sigma_1$ and denoted by $x^R$.
Given a language $A$, the notation $A^R$ denotes the \emph{reverse language} $\{x^R\mid x\in A\}$. For a family $\FF$ of languages, $\FF^R$ expresses the collection of $A^R$ for any language $A$ in $\FF$.

For any two languages $L_1$ and $L_2$, the \emph{concatenation} $L_1\cdot L_2$ (more succinctly, $L_1L_2$) denotes $\{xy \mid x\in L_1,y\in L_2\}$. In particular, when $L_1$ is a singleton $\{a\}$,  we write $aL_2$ instead of $\{a\}\cdot L_2$. Similarly, when $L_2=\{a\}$, we use the succinct notation $L_1a$.

Any function $h:\Sigma_1\to\Sigma_2^*$ for two alphabets $\Sigma_1$ and $\Sigma_2$ is called a \emph{homomorphism}. Such a homomorphism $h$ is called \emph{$\lambda$-free} if $h(\sigma)\neq\lambda$ holds for any $\sigma\in\Sigma_1$. We naturally expand a homomorphism $h$ to a map from $\Sigma_1^*$ to $\Sigma_2^*$ by setting $h(\sigma_1\sigma_2\cdots \sigma_n) = h(\sigma_1)h(\sigma_2)\cdots h(\sigma_n)$ for any $\sigma_1,\sigma_2,\ldots,\sigma_n\in\Sigma_1$.

\subsection{One-Way Probabilistic Pushdown Automata and Their Variants}\label{sec:machine-model}

As a fundamental computation model, we begin with \emph{one-way probabilistic pushdown automata} (or 1pda's, for short) as
a basis to introduce a new model of probabilistic $k$-limited automata in Section \ref{sec:formal-definition}.
\emph{One-way deterministic and nondeterministic pushdown automata} (abbreviated as 1dpda's and 1npda's, respectively) can be viewed as special cases of the following \emph{one-way probabilistic pushdown automata} (or \emph{1ppda's}, for short). We also obtain \emph{one-way unambiguous pushdown automata} (or 1upda's) as a restriction of 1npda's.

An input string is initially placed on an input tape, surrounded by two endmarkers $\cent$ (left endmarker) and $\dollar$ (right endmarker). For various models of pushdown automata, the use of such endmarkers is nonessential (see, e.g., \cite{Yam21}).
To clarify the use of the endmarkers, however, we explicitly include them in the description of a 1ppda.

Formally, a 1ppda $M$ is a tuple $(Q,\Sigma,\{\cent,\dollar\},\Gamma,\Theta_{\Gamma}, \delta,q_0,\bot,Q_{acc},Q_{rej})$, in which $Q$ is a finite set of (inner) states, $\Sigma$ is an input alphabet, $\Gamma$ is a stack alphabet, $\Theta_{\Gamma}$ is a finite subset of $\Gamma^*$ with $\lambda\in \Theta_{\Gamma}$, $\delta$ is a \emph{probabilistic transition function} from $(Q-Q_{halt})\times \check{\Sigma}_{\lambda} \times\Gamma \times Q\times  \Theta_{\Gamma}$ to $[0,1]$, $q_0$ ($\in Q$) is an initial state, $\bot$ ($\in\Gamma$) is a bottom marker, $Q_{acc}$ ($\subseteq Q$) is a set of accepting states, and $Q_{rej}$ ($\subseteq Q$) is a set of rejecting states,  $\check{\Sigma} =\Sigma\cup\{\cent,\dollar\}$,  $\check{\Sigma}_{\lambda} =\check{\Sigma}\cup\{\lambda\}$, and $Q_{halt} = Q_{acc}\cup Q_{rej}$.
Let $\Gamma^{(-)}=\Gamma-\{\bot\}$. For a given set $\Gamma$ of symbols, a \emph{$\Gamma$-symbol} refers to any symbol in $\Gamma$.
The \emph{push size} of a 1ppda is the maximum length $e$ of any string pushed into a stack by any single move, and thus $\Theta_{\Gamma}\subseteq \Gamma^{\leq e}$ follows.

For clarity reason, we express $\delta(q,\sigma,a,p,u)$ as $\delta(q,\sigma,a \mmid p,u)$ with the use of a special separator ``$|$''.
This value $\delta(q,\sigma,a \mmid p,u)$ expresses the probability that, when $M$ scans $\sigma$ on the input tape in inner state $q$, $M$ changes $q$ to $p$ and updates a topmost stack symbol $a$ to $u$.
In particular, when we always demand $\delta(q,\sigma,a \mmid p,u)\in\{0,1\}$ (instead of the unit real interval $[0,1]$) for all tuples $(q,\sigma,a,p,u)$, we obtain a \emph{one-way deterministic pushdown automaton} (or a 1dpda). In contrast, when $\delta(q,\sigma,a \mmid p,u)$ is required to be either the fixed constant or $0$, we obtain a \emph{one-way nondeterministic pushdown automaton} (or a 1npda).

For any $(q,\sigma,a)\in (Q-Q_{halt})\times \check{\Sigma}_{\lambda}\times \Gamma$, we set $\delta[q,\sigma,a] = \sum_{(p,u)\in Q\times \Theta_{\Gamma}} \delta(q,\sigma,a \mmid p,u)$. In the case of $\sigma=\lambda$, we specifically call its transition a \emph{$\lambda$-move} (or a \emph{$\lambda$-transition}) and the tape head must stay still.

At any point,  $M$ can probabilistically select either a $\lambda$-move or a non-$\lambda$-move.  This is formally stated as  $\delta[q,\sigma,a]+\delta[q,\lambda,a]=1$ for any given triplet  $(q,\sigma,a)\in (Q-Q_{halt})\times \check{\Sigma} \times \Gamma$.

Whenever $M$ reads a nonempty input symbol, the tape head of $M$ must move to the right. During $\lambda$-moves, nevertheless, the tape head must stay still. After reading $\dollar$, the tape head is considered to move off the \emph{input region}, which is marked as $\cent{x}\dollar$ for input $x\in\Sigma^*$. All cells on the input tape are indexed by natural numbers from left to right, where cell $0$ is the start cell containing $\cent$ and cell $|x|+1$ contains $\dollar$ for each input $x$.

Throughout this paper, we express a \emph{stack content}, which is a string stored inside the stack in a sequential order from the topmost symbol $w_n$ to the bottom symbol
$w_0=\bot$, as $w_n\cdots w_2w_1w_0$.

A \emph{(surface) configuration} of $M$ on input $x$ is a triplet $(q,i,w)$, which indicates that $M$ is in inner state $q$, its tape head scans the $i$th cell, and its stack contains $w$. The \emph{initial configuration} is $(q_0,0,\bot)$. An \emph{accepting} (resp., a \emph{rejecting}) \emph{configuration} is a configuration with an accepting (resp., a rejecting) state and a \emph{halting configuration} is either an accepting or a rejecting configuration. We say that a configuration $(p,j,uw)$ \emph{follows} $(q,i,au)$ with probability $\delta(q,\sigma,a \mmid p,u)$ if $\sigma$ is the $i$th input symbol and $j=i$ if $\sigma=\lambda$ and $j=i+1$ if $\sigma\neq\lambda$. A \emph{computation path} of length $k$ on the input $x$ is a series of $k$ configurations, which describes a history of consecutive transitions (or moves) made by $M$ on $x$, starting at the initial configuration with probability $p_1=1$  and the $i+1$st configuration follows the $i$th configuration with probability $p_i$ and ending at a halting configuration with probability $p_{k}$.
The probability of each computation path is determined by the multiplication of all chosen transition probabilities along the computation path.
We thus assign the probability $\Pi_{i\in[k]} p_i$ to such a computation path. It is important to note that, after reading $\dollar$, $M$ is allowed to make a finite series of $\lambda$-moves until it enters halting states. A computation path is called \emph{accepting} (resp., \emph{rejecting}) if the path ends with an accepting configuration (resp., a rejecting configuration).


Generally, a 1ppda may produce an extremely long computation path or even an infinite computation path.
Following an early discussion in Section \ref{sec:ppda-intro} on the expected runtime of probabilistic machines, it is desirable to restrict our attention to 1ppda's whose \emph{computation paths have a polynomial length on average}; that is, there is a polynomial $p$ for which the expected length of all terminating computation paths on input $x$ is bounded from above by $p(|x|)$.
A standard definition of 1dpda's and 1npda's does not have such a runtime bound, because we can easily convert those machines to ones that halt within $O(n)$ time (e.g., \cite{HU79}).
Throughout this paper, we implicitly assume that all 1ppda's should satisfy this expected polynomial termination requirement. This makes it possible for us to mostly concentrate on polynomial-length computation  paths.


Given an arbitrary string $x\in\Sigma^*$, the \emph{acceptance probability} of $M$ on $x$ is the sum of all probabilities of accepting computation paths of $M$ starting with $\cent{x}\dollar$ written on the input tape and we notationally express by $p_{M,acc}(x)$ the acceptance probability of $M$ on $x$. Similarly, we define $p_{M,rej}(x)$ to be the \emph{rejection probability} of $M$ on $x$. We say that $M$ \emph{accepts (resp., rejects) $x$ with probability $p$} if the value $p$ matches $p_{M,acc}(x)$ (resp., $p_{M,rej}(x)$).
If $M$ is clear from the context, we often omit script ``$M$'' entirely and write, e.g., $p_{acc}(x)$ instead of $p_{M,acc}(x)$.
We say that $M$ \emph{accepts} $x$ if $p_{M,acc}(x)>1/2$ and \emph{rejects} $x$ if $p_{M,rej}(x)\geq 1/2$.
Given a language $L$, in general, we say that $M$ \emph{recognizes} $L$ if, for any $x\in L$, $M$ accepts $x$ and, for any $x\in\Sigma^*-L$, $M$ rejects.
The notation $L(M)$ stands for the set of all strings $x$ accepted by $M$; that is,  $L(M)=\{x\in\Sigma^* \mid p_{M,acc}(x)>1/2\}$.
The \emph{error probability of $M$ on $x$ for $L$} refers to the probability that $M$'s outcome is different from $L$.
We further say that $M$ makes \emph{bounded error} if there exists a constant $\varepsilon\in[0,1/2)$ (called an \emph{error bound}) such that, for every input $x$, either $p_{M,acc}(x)\geq 1-\varepsilon$ or $p_{M,rej}(x)\geq 1-\varepsilon$. With or without this condition, $M$ is said to make \emph{unbounded error}.
Moreover, we say that $M$ makes \emph{one-sided error} if, for all strings $x$, either $p{M,acc}(x)>\frac{1}{2}$ or $p_{M,rej}(x)=1$ holds.

We require every 1ppda to run \emph{in expected polynomial time}.
Two 1ppda's $M_1$ and $M_2$ are \emph{(recognition) equivalent} if $L(M_1)=L(M_2)$. More strongly, we say that two 1ppda's are \emph{error-equivalent} if they are recognition equivalent and their error probabilities coincide with each other on all inputs.

For any error bound $\varepsilon\in[0,1/2)$, the notations  $\oneppda_{\varepsilon}$ and $\twopfa_{\varepsilon}$ refer to the families of all languages recognized by (expected-polynomial-time)  $\varepsilon$-error 1ppda's and (expected-polynomial-time) $\varepsilon$-error 2pfa's, respectively.
As a restriction of $\twopfa_{\varepsilon}$, $\tworfa_{\varepsilon}$ denotes the family of all languages recognized by 2pfa's with one-sided error probability at most $\varepsilon$. Similarly, we define $\onerpda_{\varepsilon}$ as the one-sided-error variant of $\oneppda_{\varepsilon}$.
In addition, we often use more familiar notation of $\pcfl$, $\bpcfl$, and $\rcfl$ respectively for $\oneppda_{ub}$, $\bigcup_{0\leq\varepsilon<1/2}\oneppda_{\varepsilon}$, and $\bigcup_{0\leq \varepsilon<1}\onerpda_{\varepsilon}$.
The strengths and weaknesses of $\pcfl$ were discussed earlier by Macarie and Ogihara \cite{MO98} and those of $\bpcfl$ were studied by  Hromkovi\v{c} and Schnitger \cite{HS10} and Yamakami \cite{Yam17}.

In comparison, $\reg$ denotes the family of all \emph{regular languages}, which are recognized by one-way deterministic finite automata. Similarly, $\cfl$ and $\dcfl$ denote the families of all languages recognized by 1npda's and by 1dpda's, respectively. It  follows that $\dcfl=\oneppda_{0}$.
Notice that $\dcfl\subseteq \rcfl \subseteq \bpcfl \subseteq\pcfl$.
Since the language $\{a^nb^nc^n\mid n\geq0\}$ is in $\bpcfl$ \cite{HS10},
$\bpcfl\nsubseteq \cfl$ follows, further leading to $\pcfl\neq \cfl$.
Furthermore, \emph{unambiguous computation} refers to nondeterministic computation consisting of at most one accepting computation path. Let us define  $\mathrm{UCFL}$ by restricting 1npda's used for $\cfl$ to produce only unambiguous computation (see, e.g., \cite{Yam14b}).

To describe the behaviors of a stack, we use the basic terminology used in \cite{Yam08,Yam14c,Yam16}. A \emph{stack content} formally means a series $z=z_mz_{m-1}\cdots z_1z_0$ of stack symbols in $\Gamma$, which are stored in the stack sequentially from $z_m$ at the top of the stack to $z_0$ ($=\bot$) at the bottom of the stack. We refer to a stack content obtained just after the tape head scans and moves off the $i$th tape cell as a  \emph{stack content at the $i$th position}. A \emph{stack transition} means the change of a stack content  by an application of a  single move.

For two language families $\FF$ and $\GG$, the notation $\FF\vee \GG$ (resp., $\FF\wedge \GG$) denotes the \emph{2-disjunctive closure} $\{A\cup B\mid A\in\FF,B\in\GG\}$ (resp., the \emph{2-conjunctive closure} $\{A\cap B\mid A\in\FF,B\in\GG\}$).
For any index $d\in\nat^{+}$, define $\FF(1)=\FF$ and $\FF(d+1)=\FF\wedge \FF(d)$. Notice that $\dcfl(k)\neq \dcfl(k+1)$ for any $k\in\nat^{+}$ \cite{LW73} (reproven in \cite{Yam20} by a different argument).

\subsection{Ideal Shape Lemma for Pushdown Automata}\label{sec:ideal-shape-lemma}

We start with restricting the behaviors of 1ppda's without compromising their language recognition power. Any 1ppda that makes such a restricted behavior is
called ``ideal shape'' in \cite{Yam21}.

We want to show how to convert any 1ppda to a ``push-pop-controlled'' form, in which (i) the pop operations  always take place by first reading an input symbol $\sigma$ and then making  a series (one or more) of the pop operations without reading any further input symbol and (ii) push operations add single symbols without altering any existing stack content.
In other words, a 1ppda \emph{in an ideal shape} is restricted to take only the following actions.
(1) Scanning $\sigma\in\Sigma$, preserve the topmost stack symbol (called a \emph{stationary operation}).  (2) Scanning  $\sigma\in\Sigma$,  push a new symbol $u$ ($\in\Gamma^{(-)}$) without changing any symbol stored in the stack. (3) Scanning  $\sigma\in\Sigma$, pop the topmost stack symbol. (4) Without scanning an input symbol (i.e., $\lambda$-move), pop the topmost stack symbol. (5) The stack operation (4) comes only after either (3) or (4).
These five conditions can be stated more formally.
We say that a 1ppda $N$ is \emph{in an ideal shape} if it satisfies the following conditions.  If $\delta(q,\sigma,a \mmid p,u)\neq0$, then (i) $\sigma=\lambda$ implies $u=\lambda$ and (ii) $\sigma\neq\lambda$ implies $u\in\{\lambda,ba,a\}$ for a certain $b\in\Gamma^{(-)}$. Moreover, for any $(q,\sigma,a)$ with $\sigma\neq\lambda$, (iii) if $\delta(q,\sigma,a \mmid p,ba)\neq0$ with $b\in\Gamma$, then $\delta[p,\lambda,b]=0$ and (iv)  if $\delta(q,\sigma,a \mmid p,a)\neq0$, then $\delta[p,\lambda,a]=0$.

Lemma \ref{transition-simple} states that any 1ppda can be converted into its ``equivalent'' 1ppda in an ideal shape.
The lemma was first stated in \cite{Yam21} for 1ppda's equipped with the endmarkers as well as the model of 1ppda's without  endmarkers. Note that \emph{1ppda with no endmarker} is obtained from the definition of 1ppda given in Section \ref{sec:machine-model} simply by removing $\cent$ and $\dollar$. The acceptance and rejection of such a no-endmarker 1ppda is determined by whether the 1ppda is in accepting states or non-accepting states just after reading off the entire input string.

\begin{lemma}{\rm [Ideal Shape Lemma, \cite{Yam21}]}\label{transition-simple}
Let $n\in\nat^{+}$. Any $n$-state 1ppda $M$ with stack alphabet size $m$ and push size $e$ can be converted into another error-equivalent 1ppda $N$ in an ideal shape with $O(en^2m^2(2m)^{2enm})$ states and stack alphabet size $O(enm(2m)^{2enm})$. The above statement is also true for the model with no endmarker.
\end{lemma}

Since error probability can differ according to inputs, by setting it appropriately, we can obtain the ideal shape lemma for 1dpda's and 1npda's.
The proof of Lemma \ref{transition-simple} given in \cite[Section 4.2]{Yam21} is lengthy, consisting of a series of transformations of automata, and is proven by utilizing, in part, basic ideas of Hopcroft and Ullman \cite[Chapter 10]{HU79} and of Pighizzini and Pisoni \cite[Section 5]{PP15}.
For completeness of the paper, we describe a rough sketch of the proof given in \cite[Section 4.2]{Yam21}.

\begin{proofsketchof}{Lemma \ref{transition-simple}}
Let us begin the proof sketch by fixing a 1ppda $M = (Q,\Sigma,\{\cent,\dollar\},\Gamma, \Theta_{\Gamma},\delta,q_0,\bot,Q_{acc},Q_{rej})$ arbitrarily. Let $|Q|=n$, $|\Gamma|=m$, and $e$ be the push size of $M$. Starting with this machine $M$, we will perform  a series of conversions of the machine to satisfy the desired condition of the ideal shape 1ppda.  To make our description simpler, we introduce the succinct notation $\delta^*[p,\lambda,a \mmid q,\lambda,w]$ ($a\in\Gamma$ and $w\in\Gamma^*$) to denote the probability of the event that, starting with state $q$ and stack content $az$ (for an arbitrary string $z\in\Gamma^+$), $M$ makes a (possible) series of consecutive $\lambda$-moves without accessing any symbol in $z$ and eventually enters state $p$ with stack content $wz$.

(1) We first convert the original 1ppda $M$ to another error-equivalent 1ppda, say,  $M_1$ whose $\lambda$-moves are restricted only to pop operations; namely, $\delta_1(q,\lambda,a \mmid p,w)=0$ for all elements $p,q\in Q$, $a\in\Gamma$, and $w\in\Gamma^+$.
For this purpose, we need to remove any $\lambda$-move by which $M$ changes topmost stack symbol $a$ to a certain nonempty string $w$. We also need to remove the transitions of the form $\delta(q,\lambda,Z_0|p,Z_0)$, which violates the requirement of $M_1$ concerning pop operations.
Notice that, once $M$ reads $\dollar$, it makes only $\lambda$-moves only with states in $Q^{(\dollar)}$ and eventually empties the stack.

(2) We next convert $M_1$ to another error-equivalent 1ppda $M_2$ that conducts only the following types of moves:
($a$) $M$ pushes one symbol without changing the exiting stack content, ($b$) it replaces the topmost stack symbol by a (possibly different) single symbol, and
($c$) it pops the topmost stack symbol.
We also require that all $\lambda$-moves of $M_2$ are limited only to either ($b$) or ($c$).

(3) We further convert $M_2$ to $M_3$ so that $M_3$ satisfies (2) and, moreover, there is no operation that replaces any topmost symbol with a different single symbol (namely, stationary operation). This is done by remembering the topmost stack symbol without writing it down into the stack. For this purpose, we use a new symbol of the form $[q,a]$ (where $q\in Q_2$ and $a\in\Gamma_2$) to indicate that $M_2$ is in state $q$, reading $a$ in the stack.

(4) We convert $M_3$ to $M_4$ that satisfies (3) and also makes only $\lambda$-moves of pop operations that only follow a (possible) series of pop operations.
Let $Q_4=Q_3\cup\{p'\mid p\in Q_{halt}\}$ and $\Gamma_4=\Gamma_3$. Let $q_{4,0}=q_{3,0}$ and $Z_{4,0}=Z_{3,0}$.
It follows that $|Q_4|\leq 2|Q_3|\leq 48en^2m^2(2m)^{2enm}$ and $|\Gamma_4|\leq 4enm(2m)^{2enm}$.
The probabilistic transition function $\delta_4$ is constructed as follows. A basic idea of our construction is that, when $M_3$ makes a pop operation after a certain non-pop operation, we combine them as a single move.

(5) Finally, we set $M_4$ as the desired 1ppda $N$.
\end{proofsketchof}

The ideal shape lemma is useful for simplifying certain proofs associated with 1ppda's. One such example was exhibited in \cite{Yam21}.

\begin{lemma}{\rm \cite{Yam21}}\label{PCFL-close-reversal}
$\pcfl$ is closed under reversal; that is, $\pcfl^{R} = \pcfl$.
\end{lemma}

\subsection{Nondeterministic Finite Automata with Output Tapes}

As done in \cite{Yam14a,Yam14b,Yam14c}, we equip each 1nfa with a \emph{write-once output tape}.\footnote{An output tape is \emph{write-once} if its cells are initially blank, its tape head never moves to the left, and the tape head must move to the right whenever it writes a non-blank symbol.}
We use such a 1nfa as a nondeterministic variant of Mealy machine, that is, a machine that produces a single output symbol whenever it processes a single input symbol.
Since the machine cannot erase written output strings, we allow the machine to  \emph{invalidate} any produced string on the output tape by later entering a rejecting inner state. For brevity, any 1nfa that behaves in this specific way is called \emph{real-time}.
Let $\rtnfamv$ denote the class of all multi-valued partial functions $f$ from $\Sigma_1^*$ to $\Sigma_2^*$ whose output values are produced on write-once tapes along only accepting computation paths of real-time 1nfa's ending in an accepting configuration in which, after the real-time 1nfa's scan the right endmarker, they enter a designated unique accepting state, where $\Sigma_1$ and $\Sigma_2$ are arbitrary alphabets. The last requirement on accepting configurations ensures that the number of all distinct accepting configurations on each input $x$ equals $|f(x)|$.
We further write $\rtnfamv_t$ for the collection of all \emph{total} functions in $\rtnfamv$.

We define the ``reversal'' $f^R$ of a function $f$ simply by setting $f^R(x) = f(x)^R$ for any $x\in\Sigma_1^*$. We use the notation  $\rtnfamv_t^R$ for $\{f^R\mid f\in \rtnfamv_t\}$. The following equality holds for the functional composition ``$\circ$'' and this equality will be used in Section \ref{sec:property-omega-LPFA}.

\begin{lemma}\label{two-rtNFAMV-collapse}
$\rtnfamv_t^R \circ \rtnfamv_t = \rtnfamv_t$.
\end{lemma}

\begin{proof}
Let $h$ denote any multi-valued total function in $\rtnfamv_t^R \circ \rtnfamv_t$. Take two functions $f,g\in\rtnfamv_t$ such that $h(x)=f(g(x)^R)$ holds for all $x$, where $g(x)^R$ denotes $\{y^R\mid y\in g(x)\}$. Note that $h(x) = \{f(y^R)\mid y\in g(x)\}$.
Take two real-time 1nfa's $M_1$ and $M_2$ with output tapes computing $f$ and $g$, respectively. For a machine $M\in\{M_1,M_2\}$, let $M = (Q_M, \Sigma, \{cent,\dollar\}, \Gamma_M,\Theta_{\Gamma_M}, \delta_M,q_{M,0},\bot, Q_{M,acc},Q_{M,rej})$.
Consider a machine that first runs $M_1$ and then runs $M_2^R$ by moving a tape head backward. This machine correctly computes $h$ but it is not a real-time 1nfa.
Here, we want to construct another real-time 1nfa $N$ that reads an input  symbol by symbol from left to right and simulates $M_1$ and $M_2^R$ simultaneously.

\s

(i) At scanning $\cent$, we guess $r_2\in Q_{acc,M_2^R}$ and store $(q_{0,M_1},r_2)$ in an internal memory. Furthermore, we guess $r_2\in Q_{M_2^R}$ and $q_2\in Q_{M_2^R}-Q_{halt,M_2^R}$ satisfying
$(r_1,\cent)\in \delta_{M_1}(q_{0,M_1},\cent)$ and $(r_2,\cent)\in\delta_{M_2^R}(q_2,\cent)$. Update the internal memory to $(r_1,q_2)$ and move to the right.

(ii) Assume that cell $\ell$ contains input symbol $\sigma\in\Sigma$ and the internal memory of $N$ contains $(q_1,r_2)$. Guess $r_1\in Q_{M_1}$ and $q_2\in Q_{M_2^R}$ satisfying $(r_1,\tau)\in\delta_{M_1}(q_1,\sigma)$ and $(r_2,\xi)\in \delta_{M_2^R}(q_2,\tau)$. Update the memory to $(r_1,q_2)$ and output $\xi$ onto $N$'s output tape.

(iii) At scanning $\dollar$, assume that $(q_1,r_2)$ is in the memory. Guess $r_1\in Q_{acc,M_1}$ satisfying  $(r_1,\dollar)\in \delta_{M_1}(q_1,\dollar)$. We then check whether  $(r_2,\dollar)\in\delta_{M_2^R}(q_{0,M_2^R},\dollar)$. If not, then $N$ enters a rejecting state. Otherwise, $N$ halts in accepting states.

\s

If there is any discrepancy in the above simulation, then our guess must be wrong, and thus $N$ immediately enters a rejecting state.

By the above construction, the number of accepting computation paths of $N$ matches that of $M$ at any input.
Therefore, we obtain $\rtnfamv_t^R \circ \rtnfamv_t \subseteq  \rtnfamv_t$.

Since the opposite inclusion $\rtnfamv_t \subseteq \rtnfamv_t^R \circ \rtnfamv_t$ is clear, the lemma is true.
\end{proof}

\section{Behaviors of Various Limited  Automata}\label{sec:observation-behavior}

We will formally introduce probabilistic models of limited automata as a foundation and explain how to obtain its variants, such as deterministic, nondeterministic, and unambiguous models. We will then  explore their fundamental properties by making structural analyses on their behaviors.

Our first goal is to provide in the field of probabilistic computation a complete characterization of finite and pushdown automata in terms of limited  automata. All probabilistic machines in this paper are assumed to run in expected polynomial time.

\subsection{Formal Definitions of Limited Automata}\label{sec:formal-definition}

In a way similar to Section \ref{sec:machine-model}, we begin with an introduction of a probabilistic model of limited automata and then define other variants by modifying this basic model.

A \emph{probabilistic $k$-limited automaton} (or a \emph{$k$-lpa}, for short) $M$ is formally defined as a tuple
$(Q,\Sigma,\{\cent,\dollar\}, \{\Gamma_i\}_{i\in[k]},\delta,q_0,Q_{acc},Q_{rej})$, which accesses
only tape area in between two endmarkers (those endmarkers can be
accessible but not changeable), where $Q$ is a finite set of (inner) states,  $Q_{acc}$ ($\subseteq Q$) is a set of accepting states, $Q_{rej}$ ($\subseteq Q$) is a set of rejecting states, $\Sigma$ is an input alphabet, $\{\Gamma_i\}_{i\in[k]}$ is a collection of mutually disjoint finite sets of tape symbols, $q_0$ is an initial state in $Q$, and $\delta$ is a probabilistic transition function from
$(Q-Q_{halt}) \times \Gamma\times Q \times \Gamma\times D$ to the real unit interval $[0,1]$ with $D=\{-1,+1\}$, $Q_{halt} = Q_{acc}\cup Q_{rej}$, and $\Gamma=\bigcup_{i\in[0,k]_{\integer}} \Gamma_i$ for $\Gamma_0=\Sigma$ and $\cent,\dollar\in \Gamma_k$.
We implicitly assume that $Q_{acc}\cap Q_{rej}=\emptyset$.
The $k$-lpa has a rewritable tape,
on which an input string is initially placed, surrounded by two endmarkers $\cent$ and $\dollar$.
In our formulation of $k$-lpa, unlike 1ppda's, the tape head always moves either to the right or to the left without stopping still. In other words, $M$ \emph{makes no $\lambda$-move}. We also remark that $M$ is not required to halt immediately after reading $\dollar$.

At any step, $M$ probabilistically chooses one of all possible transitions given by $\delta$.
For convenience, we express $\delta(q,\sigma,p,\tau,d)$ as $\delta(q,\sigma \mmid p,\tau,d)$, which indicates the probability that, when $M$ scans $\sigma$ on the tape in inner state $q$, $M$ changes its inner state to $p$, overwrites $\tau$ onto $\sigma$, and moves its tape head in direction $d$.
We set $\delta[q,\sigma]=\sum_{(p,\tau,d)\in Q\times \Gamma\times D} \delta(q,\sigma \mmid p,\tau,d)$. The function $\delta$ must satisfy  $\delta[q,\sigma]=1$ for every pair $(q,\sigma)\in Q\times\Gamma$.

We say that a tape head (sometimes its underlying machine $M$) \emph{makes a left turn} at cell $n$ if the tape head moves to the right into cell $n$ and then moves back to the left at the next step. Similarly, we define a \emph{right turn}. For convenience, the tape head is said to \emph{make a turn} if it makes either a left turn or a right turn. See Fig.~\ref{fig:head-simulation} for a tape-head movement.

The $k$-lpa $M$ must satisfy the following \emph{$k$-limitedness requirement}. During
the first $k$ scans of each tape cell, at the $i$th scan with $0\leq i<k$, if
$M$ reads the content of the
cell containing a symbol in $\Gamma_i$, then $M$ updates the cell content to another symbol in $\Gamma_{i+1}$. After the the $k$th
scan, the cell becomes unchangeable (or \emph{frozen}); that is, $M$ still
reads a symbol in the cell but $M$ no longer alters the written symbol. For this rule, there is one exception: whenever the tape head makes a turn at any tape cell, we count this move as ``double scans'' or ``double visits.''
To make the endmarkers special, we further assume that no symbol in $\bigcup_{i\in[0,k-1]_{\integer}} \Gamma_i$ replaces the  endmarkers.
The above requirement is formally stated as follows.
\begin{quote}
The $k$-limitedness
requirement: for any transition
$\delta(q,\sigma \mid p,\tau,d)\neq0$ with $p,q\in Q$, $d\in\{+1,-1\}$, $\sigma\in
\Gamma_i$, and $\tau\in\Gamma_j$ with $i,j\in[0,k]_{\integer}$, (1) if $i=k$, then
$\sigma=\tau$ and $j=k$, (2) if $i<k$ and $i$ is even, then
$j= i + 2^{(1-d)/2}$, and (3) if $i<k$ and $i$ is odd, then
$j= i + 2^{(1+d)/2}$.
\end{quote}

We assume that all tape cells are indexed by natural numbers, where the cell containing the left endmarker is indexed $0$ (such the cell is called the \emph{start cell}), the cell of the right endmarker is $n+1$ if an input $x$ is of length $n$.
Most notions and notations used for 1ppda's in Section \ref{sec:machine-model} are applied to $k$-lpda's with slight modifications if necessary. A \emph{(surface) configuration} of $M$ on input $x$ is a triplet of the form $(q,j,w)$, which indicates that $M$ is in state $q$ scanning the $j$th cell of the tape composed of $\cent w\dollar$.
Similarly to the case of 1ppda's, a \emph{computation path} of $M$ on $x$ is a series of configurations generated by applying $\delta$ repeatedly and we associate such a computation path with a probability of generating the computation path. A \emph{computation} of $M$ on $x$ relates to a \emph{computation graph} whose vertices are distinct configurations of $M$ on $x$ and (directed) edges represent single transitions of $M$ between two configurations.
Similarly to 1ppda's, we also define the notions of acceptance/rejection probability, one-sided error, bounded-error, and unbounded-error for $k$-lpa's as well as the notations, such as $p_{acc,M}(x)$ and $p_{rej,M}(x)$.
Concerning the running time of $k$-lpa's, similarly to 1ppda's in Section \ref{sec:machine-model}, we implicitly assume that all $k$-lpa's in this work \emph{run in expected polynomial time}.


When a string $w$ is written on a tape, the \emph{$w$-region} of the tape refers to a series of consecutive cells that hold each symbol of $w$, provided that $w$ can be identified uniquely on the tape from the context.
Even after some symbols of $w$ is altered, we may use the same term ``$w$-region'' as long as the referred cells in the original region are easily discernible from the context.
Moreover, a \emph{blank region} is a series of consecutive cells containing only $B$s whose ends are both adjacent to non-blank cells. A \emph{fringe} of a blank region is a non-blank cell adjacent to one of the ends of the blank region. Since two endmarkers cannot be changed, each blank region always has two fringes.
We say that $M$ \emph{enters the $w$-region in direction $d$ in inner state $q$} if, at a certain step, $M$ moves into the $d$-side of the $w$-region from the outside of $w$ by changing its inner state to $q$, where ``$d$-side'' means the left-side if $d=-1$ and the right-side if $d=+1$.   Moreover, $M$ \emph{leaves the $w$-region in direction $d$ in inner state $q$} if, at a certain step, $M$ moves out of the $d$-side of the $w$-region from the inside of $w$ by changing its inner state to $q$.


Given an index $k\geq1$ and a constant $\varepsilon\in[0,1]$, the basic notation $\klpa{k}_{\varepsilon}$ refers to the family of all languages recognized by (expected-polynomial-time) $k$-lpa's with error probability at most $\varepsilon$. In the bounded-error model, $\varepsilon$ is bounded away from $1/2$, and thus the union $\bigcup_{\varepsilon\in[0,1/2)} \klpa{k}_{\varepsilon}$ ia abbreviated as $\klbpa{k}$.
In the case of the unbounded-error model, by contrast, we write $\klpa{k}$ (occasionally, we write $\klpa{k}_{ub}$ for clarity). Similarly, $\klra{k}_{\varepsilon}$ is defined by (expected-polynomial-time) one-sided $\varepsilon$-error $k$-lpa's. Let $\klra{k} = \bigcup_{\varepsilon\in[0,1)} \klra{k}_{\varepsilon}$ and $\klra{k}_{<1/2} = \bigcup_{\varepsilon\in[0,1/2)} \klra{k}_{\varepsilon}$.

Furthermore, by requiring $k$-lna to produce only unambiguous computation on every input, we can introduce  a machine  model of \emph{unambiguous $k$-limited  automata} (or $k$-lua's, for short).
Using $k$-lda's, $k$-lna's, and $k$-lua's as underlying machines, the notations $\klda{k}$, $\klna{k}$, and $\klua{k}$ are used to express  the families of all languages recognized by $k$-lda's, $k$-lna's, and $k$-lua's, respectively.

Among all the aforementioned language families, it follows from the above definitions  that, for each index $k\geq2$,  $\klda{k}\subseteq \klua{k} \subseteq \klna{k}$, $\klda{k}\subseteq \klra{k}_{\varepsilon} \subseteq \klna{k}$, and $\klra{k}_{\varepsilon'}\subseteq \klbpa{k}_{\varepsilon'} \subseteq \klpa{k}_{\varepsilon'}$ for any constants $\varepsilon\in[0,1)$ and $\varepsilon'\in[0,1/2)$.
By amplifying the success probability of $k$-lra's, it is possible to show the further inclusion: $\klra{k} \subseteq \klbpa{k}$ for every index $k\geq1$. This inclusion is not obvious from the definition  because a $k$-lra can make error probability greater than $1/2$, which is not bounded-error probability.

\begin{lemma}\label{LRFA-to-LBPFA}
For any $k\geq1$, $\klra{k} \subseteq \klbpa{k}$.
\end{lemma}

\begin{proof}
Take any $k$-lra $M$ and assume without loss of generality that $\varepsilon \geq 1/2$. Choose a constant $\delta$ satisfying $0<\delta<\frac{1-\varepsilon}{2(1+\varepsilon)}$. We set $\alpha = 1-\frac{1-2\delta}{2\varepsilon}$. Given an input $x$, we first run $M$ on $x$. Whenever $M$ enters a rejecting state, we accept with probability $\alpha$ and reject with probability $1-\alpha$. On the contrary, when $M$ enters an accepting state, we accept with probability $1$. For any input $x\in L(M)$, the total acceptance probability becomes $1-\varepsilon+\varepsilon\cdot\alpha\geq \frac{1}{2}+\delta$. For the other input $x\notin L(M)$, the total rejection probability is $1-\alpha \geq \frac{1}{2}+\delta$. Hence, $L(M)$ belongs to $\klbpa{k}$.
\end{proof}

We further define $\klpa{\omega}$ to be the union $\bigcup_{k\in\nat^{+}} \klpa{k}$. Similarly, we can define $\klda{\omega}$, $\klra{\omega}$, $\klbpa{\omega}$, and $\klna{\omega}$. It then follows that $\klda{\omega} \subseteq \klra{\omega}\subseteq \klbpa{\omega} \subseteq \klpa{\omega}$.

\subsection{Blank Skipping Property for Limited Automata}\label{sec:blank-skipping}

Hibbard \cite{Hib67} proved that $\klna{\omega}=\cfl$  and  Pighizzini and Pisoni \cite{PP15} demonstrated that $\twolda$ coincides with $\dcfl$.
It is also possible to show that $\pcfl \subseteq \klpa{2}$ and $\bpcfl \subseteq \klbpa{2}$ using the ideal-shape property of $\pcfl$ and $\bpcfl$ (see Lemma \ref{2-lpda-construction}); however, the opposite inclusions are not known to hold.
Therefore, our purpose of exact characterizations of $\pcfl$ and $\bpcfl$ requires a specific property of $k$-lpa's, called \emph{blank skipping}, for which a $k$-lpa writes only a unique blank symbol, say, $B$ during the $k$th visit and it makes only the same deterministic moves while reading $B$ in such a way that it neither changes its inner state nor changes the head direction (either to the right or to the left); in other words, it behaves exactly in the same way while reading consecutive blank symbols.
When a $k$-lpa passes a cell for the $k$th time, it must make the cell \emph{blank} (i.e., the cell has $B$) and the cell
becomes \emph{frozen} afterward. This property plays an essential role in simulating various limited automata on pushdown automata.
In what follows, we define this property
for various limited automata.

\begin{definition}
Let $k\in\nat^{+}$. A $k$-limited automaton is said to be \emph{blank skipping} if (1) $\Gamma_k=\{\cent,\dollar,B\}$, where $B$ is a unique blank symbol, and (2) while reading a consecutive series of $B$-symbols, the machine must maintain the same inner states and the same head direction in a deterministic fashion. More formally, the condition (2) states that there are two disjoint subsets $Q_{+1},Q_{-1}$ of $Q$ for which $\delta(q,B \mid q,B,d)=1$ for any direction $d\in\{\pm1\}$ and any inner state $q\in Q_d$. See Fig.~\ref{fig:head-move}.
\end{definition}


To emphasize the use of the \emph{blank skipping property}, we append the prefix ``bs-'', as in $\mathrm{bs}\mbox{-}\klpa{k}_{\varepsilon}$ and $\mathrm{bs}\mbox{-}\klua{k}$.


\begin{figure}[t]
\centering
\includegraphics*[height=3.0cm]{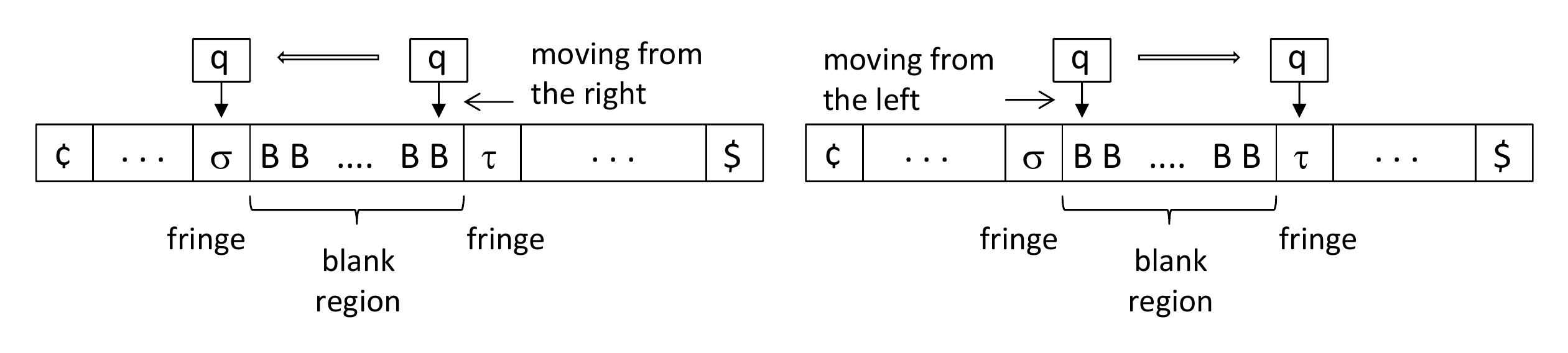}
\caption{A tape head that makes blank skipping.}\label{fig:head-move}
\end{figure}



Let us start to prove the blank skipping property of nondeterministic limited automata.

\begin{lemma}{\rm [Blank Skipping Lemma]}\label{blank-skipping}
Let $k$ be any integer with $k\geq2$. Given any $k$-lna $M$, there exists another $k$-lna $N$ such that (1) $N$ is blank-skipping, (2) $N$ is recognition-equivalent to $M$, and (3) the number of accepting computation paths of $N$ matches that of $M$ on every input.
\end{lemma}

In the case of $k$-lda's, as shown in Proposition \ref{LPFA-blank-skipping}, we can transform limited  automata into their blank skipping form and this is, in fact, a main reason that $\twolda$ equals $\dcfl$ (due to Theorem \ref{LPFA=PPDA}(2) with setting $\varepsilon=0$ and using $\dcfl=\oneppda_{0}$).
From Lemma \ref{blank-skipping}, the proposition follows immediately because the inclusions
$\mathrm{bs}\mbox{-}\klda{k} \subseteq \klda{k}$,
$\mathrm{bs}\mbox{-}\klna{k} \subseteq \klna{k}$, and $\mathrm{bs}\mbox{-}\klua{k} \subseteq \klua{k}$ are obvious and Lemma \ref{blank-skipping} yields the opposite inclusions as well.

\begin{proposition}\label{LPFA-blank-skipping}
For each index $k\geq 2$, $\klda{k}=\mathrm{bs}\mbox{-}\klda{k}$ and $\klua{k} = \mathrm{bs}\mbox{-}\klua{k}$.
\end{proposition}

In what follows, let us begin the proof of Lemma \ref{blank-skipping}.

\vs{-2}
\begin{proofof}{Lemma \ref{blank-skipping}}
The following argument uses in part a basic idea of Pighizzini and Pisoni \cite{PP15}. We first describe the proof of the lemma for $k$-lda's and then explain how to amend it for $k$-lna's (and thus for $k$-lua's).

Let $k\geq2$ be any integer and let $M = (Q,\Sigma,\{\cent,\dollar\}, \{\Gamma_i\}_{i\in[k]},\delta,q_0,Q_{acc},Q_{rej})$ be any $k$-lda.
Let $\Gamma=\bigcup_{i\in[0,k]_{\integer}}\Gamma_i$
with $\Gamma_0=\Sigma$.
Note that, as long as $a\in\Gamma-\Gamma_k$, we can uniquely determine the tape-head direction from $a$ alone. Let $D=\{\pm1\}$ and set $\ell=|Q\times D|$.

Firstly, we modify $M$ so that $M$ remembers the tape-head direction at the current step. This can be done by defining a new machine with the following items. Let $\tilde{q}_0 = (q_0,+1)$, $\tilde{Q}= Q\times D$, $\tilde{Q}_{acc}=\{(q,d)\mid q\in Q_{acc},d\in D\}$, $\tilde{Q}_{rej}=\{(q,d)\mid q\in Q_{rej},d\in D\}$, $\tilde{\delta}((q,d),a)=((p,d'),b,d')$ whenever $\delta(q,a)=(p,b,d')$. For simplicity, hereafter, we assume that $M$ has these items $\tilde{Q}$, $\tilde{Q}_{acc}$, $\tilde{Q}_{rej}$, $\tilde{q}_0$, and $\tilde{\delta}$, but we intentionally drop ``$\sim$'' (tilde) and write them as $Q$, $Q_{acc}$, $Q_{rej}$, $q_0$, and $\delta$, respectively.

Let us introduce two notations $S_{\delta}$ and $D_{\delta}$. For each string $w$, let $S_{\delta}(q,d \mmid w \mmid q',d')$ equal $1$ if $M$ enters the $w$-region in direction $d$ in inner state $q$, stays in the inside of the $w$-region, and eventually leaves the $w$-region in direction $d'$ in inner state $q'$, and let $S_{\delta}(q,d \mmid w \mmid q',d')$ be $0$ otherwise.
We also set $T_w = (t_{ij})_{i,j\in Q\times D}$ to be an $\ell\times\ell$ $\{0,1\}$-matrix such that, for any index pair $i=(q,d)$ and $j=(q',d')$ in $Q\times D$, the value $t_{ij}$ equals $S_{\delta}(i\mmid w \mmid j)$. Since the total number of distinct matrices $T_{w}$ for any $w$ is at most $2^{\ell}$, it is possible to  use $T_w$ as a part of inner states of $N$.
Moreover, let $D_{\delta}(T_u \mmid q,d \mmid T_v)$ denote the set of all pairs $(s,d')$ such that, when $d=+1$ (resp., $d=-1$), $M$ enters the $v$-region (resp., the $u$-region) in direction $-1$ (res., $+1$), stays in the $uv$-region,  and eventually leaves the $uv$-region in direction $d'$ in state $q'$. To compute the value $D_{\delta}(T_u \mmid q,d \mmid T_v)$, we need to use two matrices $T_u$ and $T_v$ but we do not need to remember $u$ and $v$.

In what follows, we wish to consider only the case of even $k$ because the other case is similar in principle. We construct the desired machine $N = (Q',\Sigma,\{\Gamma'_i\}_{i\in[k]},\delta',q'_0,Q'_{acc},Q'_{rej})$ from $M$. The desired $k$-lda $N$ works as follows.
Here, we try to meet the following requirements during the construction of $N$.
The new machine $N$ uses the same tape symbols as $M$ does, but $N$ also writes down other symbols of the form $[T_w,a,\emptyset]$, $[\emptyset,a,T_w]$, or $[T_w,a,T_u]$ to mark a fringe of a blank region. When $N$ writes $B$ over a non-blank symbol, it enters inner states of the form $[T_w,s]$ or $[T_w,s,d]$ with $s\in Q$ and $d\in D$. While $N$ stays in a blank region, however, it keeps the same inner states of the form $[T_u,s,d]$.
We first deal with the case where the tape head moves from
the left to the current cell.

\s

1. In the case of the first at most $k-2$ visits to the current cell, $N$ simulates $M$ precisely.

2. At the $(k-1)$th visit to the current cell containing symbol $a\notin \Gamma_k$, if $M$ changes the current inner state, say, $q$ to $p$ (where $p,q\in Q$) and writes  $f\in\Gamma_k$ over the symbol  $a$, then $N$ writes $B$ and changes $q$ to $[T_f,p]$ as a new inner state of $N$.

3. At the $(k-1)$th visit to $a\notin\Gamma_k$, we assume that $N$ already has inner state $[T_f,q]$, which indicates that $M$ has just written a $\Gamma_k$-symbol $f$ in the previous step over a certain non-$\Gamma_k$-symbol. Let us consider two subcases.
(a) If $M$ changes $q$ to $p$ and writes $b\notin\Gamma_k$, then $N$ enters inner state $p$ and writes $[T_f,b,\emptyset]$.
(b) In contrast, assume that $M$ writes $g\in\Gamma_k$ instead of $b$.  If $M$ moves to the right with no turn, then $N$ writes $B$,  enters inner state $[T_{fg},p]$, and moves to the right. If $M$ makes a left turn and $D_{\delta}(T_f \mmid p,-1 \mmid T_g)=(s,-1)$ (resp., $(s,+1)$), then $N$ enters $[T_{fg},s,-1]$ (resp., $[T_{fg},s,+1]$), moving to the left (resp., the right) since $(s,d)$ indicates that, once $M$ enters a region of consecutive $\Gamma_k$-symbols, it will leave this region in inner state $s$ and in direction $d$.

4. Assume that $N$ scans symbol $[\emptyset,a,T_u]$ in inner state $q$ at a fringe of a certain blank region. If $M$ moves to the right, changes $q$ to $p$, and satisfies $D_{\delta}(\emptyset \mmid q,+1 \mmid T_u)=(s,d)$, then $N$ enters inner state $[T_u,s,d]$.

5. In a blank region, $N$ maintains the same inner state of the form $[T_u,s,d]$ while moving in the same direction $d\in\{L,R\}$.

6. Consider the situation in which $N$ leaves the right end of a blank region with inner state $[T_u,s,+1]$ into its fringe that holds a symbol of the form $[T_v,b,\emptyset]$. There are three subcases to consider separately.
(a) If $M$ moves from this fringe to the right, changes the existing symbol, say, $b$ to $b'\notin\Gamma_k$, and enters inner state $p$, then $N$ updates the cell's symbol $b$ to $[T_u,b',\emptyset]$ and enters inner state $p$.
(b) In contrast, when $M$ changes $b$ to $f\in\Gamma_k$ instead, then $N$ overwrites $B$ and enters inner state $[T_{uf},p]$. In this case, the blank region grows rightwards.
(c) On the contrary, when $M$ makes a left turn, then $N$ enters inner state $[T_{uf},s',d]$ if $D_{\delta}(T_u \mmid p,-1 \mmid T_f)=(s',d)$.

7. Assume that $N$ leaves the right end of a blank region in inner state $[T_w,s,+1]$ to a non-blank cell holding a symbol of the form  $[T_u,a,T_v]$.
(a) If $M$ enters $p$ from $s$, writes $a'\notin\Gamma_k$, and moves to the right with $D_{\delta}(\emptyset \mmid p,-1 \mmid T_v)=(s',d)$, then $N$ writes $[T_w,a',T_v]$, enters inner state $[T_v,s',d]$, and moves to the right.
(b) In contrast, if $M$ writes $f\in\Gamma_k$ instead of $a'$ with $D_{\delta}(T_w \mmid q,-1 \mmid T_v)=(s',d)$, then $N$ writes $B$, enters inner state $[T_{wv},s',d]$, and moves to the right.

8. All other cases (such as $M$'s tape head moves in an opposite direction) can be similarly treated.

\s

In the case where $M$ is a $k$-lna, we wish to add an extra explanation of how to amend the aforementioned procedure to deal with any nondeterministic computation of $M$. Since $M$ makes nondeterministic choices at every step, we need to apply the above procedure to each of those choices. We also need to change the definition of $S_{\delta}$ and $D_{\delta}$ to represent ``probabilities'' of the associated events to happen. Recall that nondeterministic moves can be viewed as a special case of probabilistic moves. As one sees, if $M$ accepts an input $x$ with error probability $\varepsilon(x)$, then so does $N$, and vice versa. Therefore, every $k$-lna can be simulated by the blank-skipping $k$-lna constructed above. This completes the proof of the lemma.
\end{proofof}

\subsection{Applications of the Blank Skipping Property}\label{sec:application-BSP}

Let us explore the usefulness of the blank skipping property of various models of $1$-limited and $2$-limited automata. This rather fundamental property links such limited automata to their associated pushdown automata. More precisely, the following characterizations hold.

\begin{theorem}\label{LPFA=PPDA}
Let $\varepsilon\in[0,1/2)$ and $\varepsilon'\in[0,1)$ be any two error bounds.
\begin{enumerate}\vs{-1}
  \setlength{\topsep}{-2mm}%
  \setlength{\itemsep}{1mm}%
  \setlength{\parskip}{0cm}%

\item $\twopfa_{\varepsilon} = \onelpa_{\varepsilon}$, $\tworfa_{\varepsilon'} = \onelra_{\varepsilon'}$, and $\twopfa_{ub} = \onelpa_{ub}$.

\item $\oneppda_{\varepsilon} = \mathrm{bs}\mbox{-}\twolpa_{\varepsilon}$,  $\onerpda_{\varepsilon'} = \mathrm{bs}\mbox{-}\twolra_{\varepsilon'}$, and $\oneppda_{ub} = \mathrm{bs}\mbox{-}\twolpa_{ub}$.

\item $\dcfl = \mathrm{bs}\mbox{-}\klda{2}$, $\cfl = \mathrm{bs}\mbox{-}\klna{2}$, and $\ucfl = \mathrm{bs}\mbox{-}\klua{2}$.
\end{enumerate}
\end{theorem}

Although the equalities $\dcfl=\klda{2}$ and $\cfl=\klna{2}$ are already known \cite{Hib67,PP15}, we intensionally include them  in Theorem \ref{LPFA=PPDA}(3) in order to clarify the usefulness of the blank skipping lemma for them.

As an immediate consequence, $\rcfl$, $\pcfl$, and $\bpcfl$ are also characterized in terms of blank-skipping limited automata similarly to $\dcfl$, $\cfl$, and $\ucfl$.

\begin{corollary}\label{RCFL-BPCFL-PCFL}
$\rcfl=\mathrm{bs}\mbox{-}\klra{2}$, $\pcfl=\mathrm{bs}\mbox{-}\klpa{2}$,  and $\bpcfl = \mathrm{bs}\mbox{-}\klbpa{2}$.
\end{corollary}

Let us prove Theorem \ref{LPFA=PPDA}. Item (2) of the theorem, in particular, follows from the fact shown in Section  \ref{sec:ideal-shape-lemma} that 1ppda's can be converted into their ``ideal shapes.''
The proof of the theorem relies on two supporting lemmas: Lemmas \ref{convert-lpa-to-ppda} and \ref{2-lpda-construction}.
Lemma \ref{convert-lpa-to-ppda} helps us convert blank-skipping 2-lpa's into error-equivalent 1ppda's and Lemma \ref{2-lpda-construction} helps us convert ideal-shape 1ppda's into error-equivalent blank-skipping 2-lpda's.

\begin{lemma}\label{convert-lpa-to-ppda}
Let $n\geq2$. Every $n$-state blank-skipping 2-lpa working on an alphabet can be converted into an error-equivalent 1ppda with at most $2n$ inner states.
\end{lemma}

\begin{proof}
Given a blank-skipping 2-lpa $M$ over input alphabet $\Sigma$ with $\{\Gamma^{(1)},\Gamma^{(2)}\}$, we wish to simulate it on an appropriate 1ppda, say, $N$ with the same error probability. Starting with input $x$, when $M$ reads a new input symbol $\sigma\in\Sigma$ and changes it to $\tau$ in a single step, we read $\sigma$ and push $\tau$ into a stack. We remember a tape-head direction of $M$ internally (using inner states). In contrast, when $M$ moves its tape head leftwards by skipping the blank symbols $B$ to the first non-blank symbol, say, $\tau\in \Gamma^{(1)}$ and changes it to $B$, we simply pop a topmost stack symbol.

More formally, let $M=(Q,\Sigma,\{\cent,\dollar\}, \{\Gamma^{(1)}, \Gamma^{(2)}\}, \delta,q_0,Q_{acc},Q_{rej})$ denote any $n$-state blank-skipping 2-lpa with error probability $\varepsilon(x)$ for every input $x$ on a fixed alphabet $\Sigma$.
We define the desired 1ppda $N=(Q',\Sigma,\{\cent,\dollar\},\Gamma', \Theta_{\Gamma'},\delta',q'_0,Q'_{acc},Q'_{rej})$ in the following way. Firstly, we set $Q'=\{q^{(r)} \mid q\in Q, r\in\{+,-\}\}$, $Q'_{acc}=\{q^{(r)}\mid q\in Q_{acc},r\in\{+,-\}\}$, and $Q'_{rej}=\{q^{(r)}\mid q\in Q_{rej},r\in\{+,-\}\}$, where ``$+$'' and ``$-$'' respectively express the tape-head directions ``$+1$'' and ``$-1$'' of $\delta$.
We further set the initial state $q'_0$ to be $q^{(+)}_0$.
The desired transition function $\delta'$ is described in the following table.

Let $a\in\Gamma'$, $\sigma^{(1)}\in\Gamma^{(1)}$, and $p^{(*)}\in\{p^{(+)},p^{(-)}\}$ for any $p\in Q$. We also write $\sigma$ in $\Sigma$ as $\sigma^{(0)}$ for clarity.

\s
\begin{tabular}{ll}
$\delta'(q^{(+)},\sigma^{(0)},a \mmid p^{(+)},\sigma^{(1)}a) = \delta(q,\sigma^{(0)} \mmid p,\sigma^{(1)},+1)$ &   \\
$\delta'(q^{(+)},\sigma^{(0)},a \mmid p^{(-)},\sigma^{(1)}a) = \delta(q,\sigma^{(0)} \mmid p,B,-1)$  &  \\
$\delta'(q^{(*)},\lambda,\sigma^{(1)} \mmid p^{(+)},\lambda) = \delta(q,\sigma^{(1)} \mmid p,B,+1)$ & $\delta'(q^{(-)},\lambda,\sigma^{(1)}\mmid p^{(-)},\lambda) = \delta(q,\sigma^{(1)}\mmid p,B,-1)$ \\
$\delta'(q^{(+)},\lambda,a \mmid p^{(+)},a) = \delta(q,B \mmid p,B,+1)$ & $\delta'(q^{(-)},\lambda,a \mmid p^{(-)},a) = \delta(q,B \mmid p,B,-1)$ \\
$\delta'(q^{(+)}_0,\cent,\bot \mmid p^{(+)},\bot) = \delta(q_0,\cent \mmid p,\cent,+1)$ &
$\delta'(q^{(+)},\dollar,a \mmid p^{(-)},a) = \delta(q,\dollar \mmid p,\dollar,-1)$ \\
\end{tabular}
\s

\n By the definition of $\delta'$, it is obvious that the acceptance/rejection probability of $N$ matches that of $M$.
Moreover, the state complexity of $N$ is exactly $|Q'|=2|Q|=2n$.
\end{proof}

The ideal-shape form of 1ppda's is a key to Lemma \ref{2-lpda-construction}.

\begin{lemma}\label{2-lpda-construction}
Let $n,l\in\nat^{+}$. Let $L$ be a language over alphabet $\Sigma$ of size $l$ recognized by an $n$-state 1ppda $M$ in an ideal shape with error probability $\varepsilon(x)$ for any input $x$. There is a blank-skipping $2$-lpa $N$ that has $O(nl)$ states and recognizes $L$
with the same error probability.
The average runtime of $N$ is $O(|x|^2)$ multiplied by the average runtime of $M$ for any input $x$.
\end{lemma}

\begin{proof}
Let $n$ and $l$ be any two constants in $\nat^{+}$ and let $\Sigma$ denote any alphabet of size $l$.
We take a language $L\subseteq\Sigma^*$ and an ideal-shape 1ppda $M = (Q,\Sigma,\{\cent,\dollar\},\Gamma, \delta,q_0, \bot,Q_{acc},Q_{rej})$ that recognizes $L$ with  error probability $\varepsilon(x)$ for any input $x$.

Without loss of generality, we further assume that, along each computation path, $M$ enters a halting state only on or after reading the endmarker $\dollar$ followed by a (possible) series of $\lambda$-moves. We want to construct a $2$-lpa $N = (Q',\Sigma,\{\Gamma_i\}_{i\in [2]}, \delta',q'_0, Q'_{acc},Q'_{rej})$ to
simulate $M$ with the same error probability.
A basic idea of the simulation of $M$ proceeds in the following way. Assuming that $M$ reads a new input symbol $\sigma$, if $M$ pushes a stack symbol $a$ into its stack, then we change $\sigma$ into $a$. If $M$ pops a topmost stack symbol, then we move a tape head leftwards to read the first non-blank symbol (which corresponds to the topmost stack symbol) and then replace it with $B$ (blank symbol).
In the case where $M$ makes a $\lambda$-move, since $M$'s move must be a pop operation because of the ideal-shape form, we take the same action as stated before.
In order to make this simulation blank-skipping, we further need to remember a state-stack pair of $M$'s last step and utilize it to skip a blank region.

To implement the above basic idea, we introduce the special notation of $[qa]$ to indicate that $M$ is in inner state $q$, reading stack symbol $a$ as well as a certain nonempty input symbol.
We define $Q'$ to be $(Q\times \Sigma) \cup Q$.
A new initial state $q'_0$ is set to be $[q_0\bot]$ and we define  $\Gamma_0=\Sigma$ and
$\Gamma_{2}=\{\cent,\dollar,B\}$.
For simplicity, we use the same notation $B$ to express ``blank'' in $\Gamma_1$ as well as in $\Gamma_2$. This irregular usage does not cause any problem because we can easily transform it into a standard $2$-lpa by adding extra symbols and extra transitions.

In contrast, while $M$ makes a series of $\lambda$-moves, $N$ stays in inner states taken from $Q$. When $M$ pushes $b$ and enters state $p$ with  probability $\gamma$, $N$ changes its inner state $[qa]$ to $[pb]$ and writes $a$ over $\sigma$ with the same probability. When $M$ pops and enters state $p$ from $q$, $N$ updates $\sigma$ to $B$, changes its inner state $[qa]$ to $p$. Recall that a pop operation at a non $\lambda$-move triggers a series of pop operations at $\lambda$-moves.
Hence, we can move the tape head to the left and to delete the corresponding symbols. Moreover, $N$ always skips the blank symbol $B$ with probability $1$. It is clear that $|Q'|= |Q||\Sigma|+|\{\bar{q} \mid q\in Q\}|= nl+n\leq 2nl$.

The formal definition of $\delta'$ is given in the following table. Let $\sigma\neq\lambda$ and $\sigma\in\Sigma$.

\s
\begin{tabular}{ll}
$\delta'([qa],\sigma \mmid [pa],B,+1) = \delta(q,\sigma,a \mmid p,a)$
& $\delta'([qa],\sigma \mmid [pb],a,+1) = \delta(q,\sigma,a \mmid p,ba)$ \\
$\delta'([qa],\sigma \mmid p,B,-1) = \delta(q,\sigma,a \mmid p,\lambda)$ ($a\neq \bot$)
& $\delta'(q,a \mmid p,B,-1) = \delta(q,\lambda,a \mmid p,\lambda)$ ($a\neq \bot$) \\

$\delta'(q,a \mmid [qa],B,+1) =1$ whenever $\delta[q,\lambda,a]=0$
& \\

$\delta'([qa],B \mmid [qa],B,+1) = 1$ & $\delta'(q,B \mmid q,B,-1) = 1$  \\

$\delta'([qa],\dollar \mmid q,\dollar,-1)= \delta(q,\dollar,a \mmid q,\lambda)$ ($a\neq \bot$)
&    \\
$\delta'(q'_0, \cent \mmid [p\bot],\cent,+1) = \delta(q_0,\cent,\bot \mmid p,\bot)$
& $\delta'(q'_0, \cent \mmid [pa],\cent,+1) = \delta(q_0,\cent,\bot \mmid p,a\bot)$.
\end{tabular}
\s

\n Note that $\bar{q}$ and $\bar{p}$ are new inner states not in $Q$, which respectively correspond to $q$ and $p$ and that $q_h$ is a member of $Q_{halt}$.
The above transitions indicate that the new 2-lpa satisfies all the requirements stated in the proposition.

Concerning the runtime of $N$, we note that the tape head always advances toward $\dollar$, except that, to erase a symbol in $\Gamma_1$, it moves backwards, changes the first encountered non-blank symbol to $B$, and returns. After reading $\dollar$, the tape head simply moves to the left to $\cent$. Thus, the average number of steps is upper bounded by $O(|x|^2)$ multiplied by the average runtime of $M$ on $x$.
\end{proof}

We then return to Theorem \ref{LPFA=PPDA} and provide its proof, based on Lemmas \ref{convert-lpa-to-ppda} and \ref{2-lpda-construction}.

\vs{-2}
\begin{proofof}{Theorem \ref{LPFA=PPDA}}
(1) It is rather easy to simulate a 2pfa on a 1-lpa that behaves like the 2pfa but changes each input symbol $\sigma$ to its corresponding symbol $\sigma'$.
On the contrary, we can simulate a 1-lpa $M$ by the following  2pfa $N$. A key idea of this simulation is that it is possible to transform $M$ into the specific form where $M$ always rewrites a new corresponding symbol over each distinct input symbol. Since $N$ can remember the correspondence between input symbols and their corresponding new symbols, whenever $N$ reads any overwritten symbol $\tau$, $N$ can translate it back to its  original input symbol $\sigma$ and then simulate $M$'s single step of processing $\sigma$.

(2) These equalities directly come from Lemmas \ref{convert-lpa-to-ppda} and \ref{2-lpda-construction}. Here, we only show that $\oneppda_{\varepsilon} = \mathrm{bs}\mbox{-}\klpa{2}_{\varepsilon}$ since the other equalities can be proven similarly.

Let $L$ be any language in $\oneppda_{\varepsilon}$ and take a 1ppda $M$ that recognizes $L$ with error probability at most $\varepsilon\in[0,1/2)$. By the ideal shape lemma (Lemma \ref{transition-simple}), $M$ can be assumed to be in an ideal shape. We then apply Lemma \ref{2-lpda-construction} and obtain a blank-skipping $2$-lpa $N$ that simulates $M$ with the same error probability on all inputs. Therefore, $L$ belongs to $\klpa{2}_{\varepsilon}$.
Conversely, for any language $L$ in $\mathrm{bs}\mbox{-}\klpa{2}$, take a blank-skipping $2$-lpa $N$ satisfying $L=L(N)$. Lemma \ref{convert-lpa-to-ppda} implies the existence of an ideal-shape 1ppda $M$ that is error-equivalent to $N$. We thus conclude that $L$ is in $\oneppda_{\varepsilon}$.

(3) This can be proven in a similar way as (2).
\end{proofof}

\subsection{Power of Limited Randomness}\label{sec:example-language}

We will see how randomness endows enormous power to underlying limited automata. In automata theory, the computational power of randomness was first observed by Freivalds \cite{Fre81b}, who designed a two-way probabilistic finite automaton that recognizes the non-regular language $\{a^nb^n\mid n\in\nat\}$ with bounded-error probability.

The language family $\klra{k}$ turns out to be relatively large. This comes from the fact that $\klpa{k}$ contains languages not recognized by any $k$-lda for each fixed index $k\geq2$.

\begin{theorem}\label{omega-lda-vs-2lrfa}
For any index $k\geq2$, $\klda{k} \subsetneqq \klra{k} \cap \klda{(k+1)}$.
\end{theorem}

Unfortunately, Theorem \ref{omega-lda-vs-2lrfa} is not strong enough to yield the separation of $\klda{\omega}\neq \klra{\omega}$. We also do not know whether or not $\klda{3}\nsubseteq \twolra$ as well as  $\klra{\omega}\nsubseteq \twolbpa$.


Let us state the proof of Theorem \ref{omega-lda-vs-2lrfa}. The proof  utilizes example languages, which can separate $\klda{k}$ from $\klda{(k+1)}$  for each index $k\geq2$.
Earlier, Hibbard \cite{Hib67} devised such an example language for each index $k\geq2$. In the case of $k=2$, a much simpler example language was given in \cite{PP15}.
The example language $L_{k}$ given below is, in fact, a slight modification of the one stated in \cite{Hib67}.
An argument similar to \cite[Section 4]{Hib67} verifies that,
for each index $k\geq2$, the language $L_{k}$ is included in $\klda{(k+1)}$ but outside of $\klda{k}$.

\s

(1) When $k=2$,  $L_2$ denotes the language $\{a^nb^nc,a^nb^{2n}d\mid n\geq0\}$.

(2) Let $k\geq3$ be any integer. We succinctly write $w_1$ for either $a$ or $b$ and, for any index $i\in[2,k-1]_{\integer}$, we write $w_i$ for $a^{n_i}b^{m_i}c^{p_i}$ for arbitrary numbers $n_i,m_i,p_i\in\nat$.
The desired language $L_k$ is composed of all strings $w$ of the form $w_2\#w_4\#\cdots \# w_{k-2} \# w_{k}\# w_{k-1} \# \cdots \#w_5\#w_3\# w_1$ if $k$ is even, and $w_2\#w_4\#\cdots \# w_{k-1} \# w_{k}\# w_{k-2} \# \cdots \#w_5\#w_3\# w_1$ if $k$ is odd, together with the following three conditions.

\begin{enumerate}\vs{-2}
  \setlength{\topsep}{-2mm}%
  \setlength{\itemsep}{1mm}%
  \setlength{\parskip}{0cm}%

\item[(i)] If $w_1=a$, then $n_2\leq p_2$ holds; otherwise, $m_2\leq p_2$ holds.

\item[(ii)] Let $j$ denote any index in $[3,k-1]_{\integer}$. If either $n_{j-1}=p_{j-1}$ or $n_{j-1}<m_{j-1}$, then $n_j\leq m_j$ holds. If either $n_{j-1}<p_{j-1}$ or $n_{j-1}=m_{j-1}$, then $n_j\leq p_j$ holds.

\item[(iii)] If $n_{k-1}=p_{k-1}$, then $n_{k}=m_{k}$ holds. If $n_{k-1}<p_{k-1}$, then $n_{k}<p_{k}$ holds. If $n_{k-1}=m_{k-1}$, then $n_{k}=p_{k}$. If $n_{k-1}<m_{k-1}$, then $n_{k}<m_{k}$ holds.
\end{enumerate}


\vs{-3}
\begin{proofof}{Theorem \ref{omega-lda-vs-2lrfa}}
We first show the case of $k=2$, that is, $L_2\in \klra{2}\cap \klda{3}$ but $L_2\notin \klda{2}$. Clearly, $L_2$ is in $\klda{3}$ and it is also known that $L_2\notin\dcfl$ \cite{Yu89}.
Since $\dcfl=\klda{2}$ \cite{PP15},  $L_2\notin \klda{2}$ follows.
It thus suffices to show that $L_2\in \onerpda_{1/2}$ because Theorem \ref{LPFA=PPDA} implies $L_2\in\klra{2}$.
Consider the following 1ppda. On input $x$, choose $c$ and $d$ with equiprobability $1/2$ and let $z$ denote the chosen symbol. When $z=c$, check deterministically whether $x=a^nb^nc$; in contrast, when $z=d$,  check whether $x=a^nb^{2n}d$. If this is the case, then accept $x$; otherwise, reject it. The probability that $x$ is accepted is exactly $1/2$ if $x\in L_2$. On the contrary, if $x\notin L_2$, then $x$ is rejected with probability $1$. Hence, $L_2$ belongs to $\onerpda_{1/2}$, as requested.

Let $k\geq3$ and consider $L_{k}$. Since $L_{k}\in \klda{(k+1)} - \klda{k}$ \cite{Hib67}, it suffices to prove that  $L_{k}\in \klra{k}$. We focus on the case that $k$ is odd. Take a $(k+1)$-lda $M_k$ that recognizes $L_k$. We further assume that $M_k$ first checks whether $w_1=a$ or $w_1=b$ by moving its tape head to the $w_1$-region and that, once $M_k$ leaves the $w_1$-region, the tape head never returns to this region.

Let us consider a new machine $N_k$ defined as follows. On input $w$ of the form $w_2\# w_4 \# \cdots \# w_{k-1}\# w_{k} \# w_{k-2}\# \cdots w_3\# w_1$, choose an element $z\in\{a,b\}$ with equal probability and then simulate $M_k$ from the moment when its tape head moves back from $w_1$ toward the $w_2$-region to start checking that $w_2,w_3,\ldots,w_k$ satisfy Conditions (ii)--(iii). If $M_k$ eventually rejects $w$, then $N_k$ also rejects it. Otherwise, $N_k$ moves its tape head to the $w_1$-region and check if $w_1=z$. If so, then $N_k$ accepts $w$; otherwise, $N_k$ rejects it. Since $M_k$ is deterministic,  if $w\in L_k$, then $N_k$ accepts $w$ with probability $1/2$; otherwise, $N_k$ rejects it with probability $1$. It is also clear that $N_k$ is a $k$-limited automaton. Therefore, $N_k$ is a $k$-lra and, since $L_k=L(N_k)$, $L_{k}$ belongs to $\klra{k}$.

The other case that $k$ is even is similarly treated.
\end{proofof}


If we restrict the values of transition probabilities, then the associated  bounded-error 2pfa's are, in general, significantly less powerful. As such an example,
Wang \cite{Wan92} earlier showed that $\dcfl$ contains all languages recognized with bounded-error probability by 2pfa's having \emph{rational transition probabilities}.
For notational convenience, we denote by $\klbpa{k}^{(rat)}$ the subclass of $\klbpa{k}$ defined by $k$-lpa's using only rational transition probabilities. Theorem \ref{LPFA=PPDA}(1) thus implies the following.

\begin{corollary}
$\onelbpa^{(rat)}\subseteq \dcfl \subsetneqq \twolra$.
\end{corollary}

\subsection{Collapses of $\omega$-LPFA, $\omega$-LRA,  and $\omega$-LUA}\label{sec:property-omega-LPFA}

Certain models of limited automata can reduce the number of permitted tape-symbol modifications. For instance, any $k$-lna with $k\geq2$ has its recognition-equivalent $2$-lna, leading to the collapse of   $\klna{\omega}$ down to $\klna{2}$ \cite{Hib67}.
Here, we intend to seek similar collapses for other models of limited automata. In particular, we wish to demonstrate
the collapses of $\klpa{\omega}$ to $\twolpa$, $\klra{\omega}$ to $\twolra$, and $\klua{\omega}$ to $\klua{2}$ (where $\klua{2}$ equals $\ucfl$ by Proposition \ref{LPFA-blank-skipping} and Theorem \ref{LPFA=PPDA}). This result  hints  the robustness of the language families $\twolpa$, $\klra{2}$, and $\klua{2}$.

\begin{theorem}\label{LPFA-and-PCFL}
$\klpa{\omega} = \twolpa$, $\klra{\omega} = \klra{2}$, and $\klua{\omega}=\klua{2}$.
\end{theorem}

Since $\klda{\omega}$ is contained in both $\klra{\omega}$ and $\klua{\omega}$, and $\klbpa{\omega}$ is also contained in $\klpa{\omega}$, Theorem \ref{LPFA-and-PCFL} further implies the following containments.

\begin{corollary}
$\klda{\omega}\subseteq \klra{2} \cap \klua{2}$ and $\klbpa{\omega}\subseteq \klpa{2}$.
\end{corollary}

Hereafter, we intend to prove Theorem \ref{LPFA-and-PCFL}.
We remark that an argument given in the proof of the theorem further proves Hibbard's collapse result of  $\klna{\omega}=\klna{2}$.
The proof requires three fundamental relationships between $\klpa{k}$ and $\klpa{(k+1)}$,  between $\klra{k}$ and $\klra{(k+1)}$, and between $\klua{(k+1)}$ and $\klua{k}$ regarding ``reversal'' for any index  $k\geq2$.

Let us consider any multi-valued total function  $f:\Sigma_1^*\to\Sigma_2^*$, for two alphabets $\Sigma_1$ and $\Sigma_2$, witnessed by a certain nondeterministic Turing machine (not necessarily limited to a 1nfa) $M_f$ with a write-once output tape, provided that $M_f$ moves to $\dollar$ and enters a fixed accepting state, say, $q_{acc}$ unless $M_f$ enters any rejecting state to invalidate the existing  strings produced on the output tape.
For any probabilistic Turing machine $M$ over $\Sigma_2$, we use the notation $L_{f,M}$ for the special language  $\{x\in\Sigma_1^*\mid \sum_{y\in f(x)} \prob_{M}[\text{ $M$ accepts $y$ }] > \frac{1}{2} |f(x)|\}$.
Let us recall the function class $\rtnfamv_t$ from Section \ref{sec:machine-model}.
By abusing the notation, we write $\klpa{k}\circ \rtnfamv_t$ to denote the set of all such $L_{f,M}$'s for any $f\in\rtnfamv_t$ and any $k$-lpa $M$. When $M$ is a $k$-lra, we choose an arbitrary constant $\eta\in(0,1/2)$ and define $L_{f,M,\eta}$ as  $\{x\in\Sigma_1^*\mid \sum_{y\in f(x)} \prob_{M}[\text{ $M$ accepts $y$ }] > \eta |f(x)| \}$.
In the case where $M$ is a $k$-lna, by contrast, we set $L_{f,M}$ to be $\{x\in\Sigma_2^*\mid \exists y\in f(x)[\text{ $M$ accepts $y$ }]\}$ and write $\klua{k}\circ \rtnfamv_t$ for the collection of all such $L_{f,M}$'s.

We argue that $\klpa{k}$ as well as $\klra{k}$ and $\klua{k}$ is ``invariant'' with an application of $\rtnfamv_t$-functions in the following sense.

\begin{lemma}\label{PCFL-MV-composition}
For any index $k\geq2$, $\klpa{k}\circ \rtnfamv_t = \klpa{k}$.
The same equality holds for $\klra{k}$ and $\klua{k}$ as well.
\end{lemma}

\begin{proof}
Let $k\geq2$. We first concentrate on the case of $\klpa{k}$. It is obvious that $\klpa{k}\subseteq \klpa{k}\circ \rtnfamv_t$ by considering the \emph{identity function} in $\rtnfamv_t$. Thus, we need to prove only the opposite inclusion. Take a multi-valued total function $f$ in $\rtnfamv_t$ and a $k$-lpa $M$, and consider the associated language $L_{f,M}$. There exists a real-time 1nfa $M_f$ with a write-once output tape that computes $f$.
Let us consider the following probabilistic machine $N$. Starting with input $x$, $N$ runs $M_f$ and, along each computation path of $M_f$, whenever $M_f$ produces one nonempty output symbol $\sigma$,$N$  runs  $M$ to process $\sigma$. If $M_f$ enters an accepting state in the end, then $N$ does the same after processing the last output symbol of $f$; otherwise, $N$ enters both accepting and rejecting states with equiprobability after finishing the processing of $M_f$'s output symbols.
Note that $N$ is a $k$-lpa because $M_f$ does not need any memory space and $N$ can process each output symbol of $M_f$ without using $M_f$'s  output tape.
It thus follows that $p_{acc,N}(x) = \alpha_{M}(x) \sum_{y\in f(x)}\frac{p_{acc,M}(y)}{|f(x)|}  + \frac{1}{2}(1-\alpha_M(x))$, where $\alpha_{M}(x)$ is the ratio of the number of accepting paths out of the total computation paths of $M$ on $x$. If $x\in L_{f,M}$, then we obtain $p_{acc,N}(x) >\frac{1}{2}$; otherwise, $p_{acc,N}(x)\leq \frac{1}{2}$ follows. We thus conclude that $N$ recognizes $L_{f,M}$ with unbounded-error probability. Therefore, $L_{f,M}$ belongs to $\klpa{k}$.

By slightly modifying the above argument, we can show the lemma for $\klra{k}$ and $\klua{k}$.
\end{proof}

Consider any $k$-lpa $M$ used in the aforementioned definition of $\klpa{k}\circ \rtnfamv_t$. If we feed such an $M$ with the reverses $x^R$ of inputs $x$, then we obtain $\klpa{k}^R\circ \rtnfamv_t$. We show the following key relationships between $\klpa{(k+1)}$ and $\klpa{k}^R$, between $\klra{(k+1)}$ and $\klra{k}^R$, and between $\klua{(k+1)}$ and $\klua{k}^R$.

\begin{lemma}\label{reduction-from-k+1-to-k}
For any $k\geq2$, $\klpa{k}^R\circ \rtnfamv_t = \klpa{(k+1)}$. The same equality also holds between $\klra{k}^R$ and $\klra{(k+1)}$ and  between $\klua{k}^R$ and $\klua{(k+1)}$.
\end{lemma}

\begin{proof}
Fix $k\geq2$. Firstly, we intend to verify the lemma for $\klpa{k}$ and then explain how to modify the proof for $\klra{k}$ as well as $\klua{k}$.
For readability, we split the proof into two different parts.

\s

(1) We first plan to derive the inclusion of $\klpa{(k+1)} \subseteq \klpa{k}^R\circ \rtnfamv_t$. Let $L$ be any language in $\klpa{(k+1)}$ witnessed by a $(k+1)$-lpa $M = (Q_M,\Sigma,\{\cent,\dollar\}, \{\Gamma^{(i)}_M\}_{i\in[k]}, \delta_M, q_0, Q_{acc}, Q_{rej})$. Let $\Gamma^{(0)}_M = \Sigma$.
For convenience, assume that $M$ takes $q_0$ only at the first step and that $M$ halts after its tape head scans $\dollar$ (but this is not necessarily the first time).

To simulate the behavior of $M$, we introduce a notion of \emph{cell state}. A cell state is of the form $(d,q,\sigma \mmid \tau, p, e \mmid h)$, which indicates that a tape head of $M$ enters the current tape cell, say, $\ell$ from direction $d\in\{L,R,N\}$ (where $L$ means the left, $R$ means the right, and $N$ refers to the start cell) in inner state $q$, reads tape symbol $\sigma$, changes $q$ to $p$ and $\sigma$ to $\tau$, and leaves cell $\ell$ in direction $e\in\{L,R\}$. As for  $h\in (Q^2\times \Gamma)\cup \{\lambda\}$, if
$(d,e)=(L,L)$, then $M$ makes a left turn, and thus $h$ indicates the information on $M$'s situation at the time when $M$ visits cell $\ell$ for the 3rd time. If $d\neq R$ and $e=R$, then $h$ must be $\lambda$.
Along each computation path of $M$, the first ``traverse'' of $M$'s tape head  from left to right uniquely determines a series of such cell states.


\begin{figure}[t]
\centering
\includegraphics*[height=4.0cm]{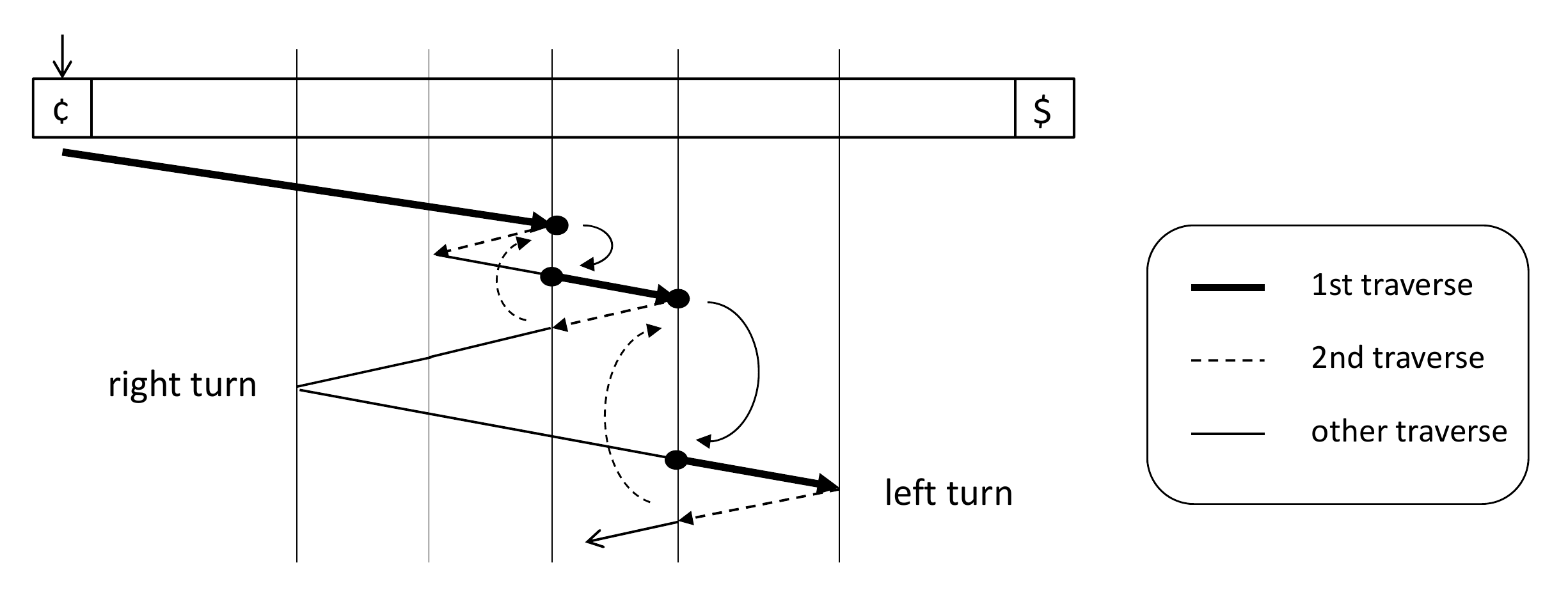}
\caption{The simulation of the first and the second traverses of a tape head of $M$.}\label{fig:head-simulation}
\end{figure}


As the first step, we want to define a multi-valued total function $g$ computed by a real-time 1nfa $M_g$ that simulates the first traverse of $M$'s tape head.
This function $g$ nondeterministically describes the ``actions'' of $M$ while it reads input symbols for the first time.
On an input, $M_g$ partly simulates $M$ symbol by symbol but $M_g$'s tape head moves only to the right. Consider tape cell, say, $\ell$ and let $\sigma$ be the input symbol in $\check{\Sigma}$ ($=\Sigma\cup\{\cent,\dollar\}$) written in cell $\ell$. For the ease of description of the construction of $M_g$, we arbitrarily fix a computation path of $M$ and explain how to operate $M_g$ along this particular computation path (as if $M$ works deterministically).
There are several cases to treat separately.
At any moment, if $M$ enters any halting state, then $M_g$ does the same to invalidate the current output strings since this situation suggest that the simulation of $M$ by $M_g$ along this particular computation path is wrong.

\s

(i) Consider the case $M$ is in initial state $q_0$ reading $\cent$. If $M$ changes $q_0$ to $p$ and leaves cell $\ell$ to the right, then $M_g$ writes down $(N,q_0,\cent\mid \cent,p,R \mmid \lambda)$ to its output tape. We store $p$ in $M_g$'s internal memory and move the tape head of $M$ to the right in inner state $p$.

(ii) Consider the case where $M$'s tape head enters cell $\ell$ for the first time from the left and reads $\sigma\in\check{\Sigma}$ in inner state $q\in Q-\{q_0\}$ but $M$ does not make any left turn. Let the memory of $M_g$ contain $q$.
If $M$ changes $q$ to $p$, changes $\sigma$ to $\tau$, and leaves  cell $\ell$ to the right, then $M_g$ writes down $(L,q,\sigma \mid \tau, p,R \mmid \lambda)$ on its output tape. We remember $p$ and move the tape head of $M$ to the right in inner state $p$.

(iii) Assume that $M$ enters cell $\ell$ for the first time from the left and makes a left turn. Assume that $q$ is in $M_g$'s memory.
In this case, we ``skip'' many subsequent moves of $M$ and ``jump'' into $M$'s step at which $M$ returns to cell $\ell$ for the second  time. See Fig.~\ref{fig:head-simulation}. We guess a triplet $(t,r,a)$ satisfying $(r,a,R)\in\delta_M(t,\tau)$. If there is no such triplet, then $M_g$ enters a rejecting state to invalidate the existing output strings. If $M$ changes $q$ to $p$, then  $M_g$ writes down $(L,q,\sigma\mmid \tau,p,L\mmid t,r,a)$ and remembers $r$. We then move the tape head of $M$ to the right in inner state $r$.

(iv) Consider the case where cell $\ell$ contains $\dollar$ and $M$ enters cell $\ell$ for the first time from the left. Assume that $q$ is in the memory of $M_g$. If $M$ changes $q$ to $p$ and moves backward, then $M_g$ produces $(L,q,\dollar \mmid \dollar, p,L \mmid \lambda)$ and halts  by entering a designated accepting state $q_{acc}$.
If $M$ processes the endmarker $\dollar$ and terminates, then $M_g$ produces $(L,q,\dollar \mmid \dollar,p,R\mmid \lambda)$ and halts in inner state $q_{acc}$.

\s

Clearly, $M_g$ is a real-time 1nfa, and therefore $g$ belongs to $\rtnfamv_t$. Given any index $\ell\in[0,n+1]_{\integer}$, we call cell $\ell$ a \emph{turning point} if $M$ makes a left turn at cell $\ell$ during the first ``traverse''.

As the second step, we wish to define a $k$-lpa $N$ that takes an input provided by $g$ and simulates $M$.
For readability, we describe the behavior of $N$ starting at the left end of the tape.
As before, we fix an arbitrary computation path of $M$ and assume that $M_g$ operates along this particular computation path associated with $g$'s outcome  $(N,q_0,\cent\mid \cent,p,R\mmid \lambda) w (L,q,\dollar\mid \dollar,p,d \mmid h)$, where $d\in\{L,R\}$ and $w$ is a (possibly empty) series of cell states.
Assume that $w$ is a length-$n$ sequence $w_1w_2\cdots w_n$ of cell states.  For convenience, we set $w_0=(N,q_0,\cent\mid \cent,p,R\mmid \lambda)$ and $w_{n+1} = (L,q,\dollar\mid \dollar,p,d \mmid h)$. The input to $N$ is thus $(w_0 w w_{n+1})^R$, which is abbreviated as $\widetilde{w}^R$.

Given $\tilde{w}^R$, $N$ starts to simulate $M$'s second ``traverse'' from left to right until the next turning point.
While $M$ simulates $M$'s second traverse, if $N$ encounters the next turning point at cell $\ell$, instead of continuing the current simulation of $M$, $N$ ``jumps'' to the time of the left turn and resumes the simulation of $M$'s moves from this point.
Even if $M$ makes a right turn at cell $\ell$ at its 2nd visit,
$N$ continues the simulation of $M$'s moves.
A more detailed simulation process is described below. Let $\ell$ denote any cell index in $[n]$.

\s

(i$'$) The machine $N$ begins in its initial state $q'_0$,
scanning $w_{n+1}=(L,q,\dollar\mid \dollar,p,L \mmid \lambda)$ written in the rightmost cell of its input tape.
If $M$ makes a transition of the form $(s,\dollar,L)\in\delta_{M}(p,\dollar)$, then $N$ writes down  $\track{w_{n+1}}{\dollar}$ into this cell,
remembers $s$ in $N$'s memory, and moves to the left.

(ii$'$) Assume that $N$'s tape head is scanning $w_{\ell} = (L,q,\sigma\mid \tau,p,R\mmid \lambda)$ with $s$ in its memory.
If $M$ follows a transition of the form $(r,a,d)\in\delta_M(s,\tau)$   along the current computation path with $d\in\{L,R\}$, then $N$ writes down $\track{w_{\ell}}{a}$ into this cell, remembers $r$, and moves in direction $d$.

(iii$'$) Assume that $\sigma\neq \dollar$ and $N$'s tape head is scanning $w_{\ell}=(L,q,\sigma \mmid \tau,p,L \mmid t,r,a)$ with $s$ in the memory. This implies that $M$ makes a left turn at cell $\ell$ during the first traverse and that $M$ returns to cell $\ell$ in inner state $t$ during the second traverse, and in inner state $s$ during the 3rd traverse.
If $M$ makes a transition ``$(u,b,L)\in\delta_{M}(s,a)$'', then  $N$  writes $\track{w_{\ell}}{(u,s,a)}$ over $w_{\ell}$, remembers $p$, and moves to the left.

(iv$'$) Assume that $N$ scans $\track{w_{\ell}}{a}$ with $s$ in the memory for symbol $a\in\Gamma^{(i)}_{M}$ with $i\in[2,k]_{\integer}$. Note that $N$ has already visited this cell at most $i-1$ times if $i<k$.
If $M$ follows a transition ``$(r,b,d)\in \delta_M(s,a)$'' with direction $d\in\{L,R\}$, then $N$ writes down $\track{w_{\ell}}{b}$ into this cell, remembers $r$, and moves in direction $d$.

(v$'$) Assume that  $N$ enters cell $\ell$ from the left and scans $\track{w_{\ell}}{(u,s,b)}$ with $z$ in the memory.
Since this is a turning point, we assume that $w_{\ell}$ has the form $(L,q,\sigma \mmid \tau,p,L\mmid t,r,a)$.
Since $z\neq t$ leads to a discrepancy, we need to consider only the case of  $z=t$. The machine $N$ then writes down $\track{w_{\ell}}{b}$, remembers $u$, and moves to the left.

(vi$'$) Assume that $N$ scans $w_0=(N,q_0,\cent \mmid \cent,p,R\mmid \lambda)$ with $s$ in the memory. If $M$ makes a transition of the form $(r,\cent,R)\in\delta_M(s,\cent)$, then $N$ writes down  $\track{w_0}{\cent}$, remembers $r$, and moves to the right.

\s

Since $g$'s outcomes contain ``guessed'' results, whenever the above simulation process discovers any discrepancy on these outcomes, $N$ immediately enters both accepting and rejecting states with equal probability. This cancels out all erroneous computations of $M_g$. When $N$ correctly terminates, it successfully simulates $M$'s computation with the help of $g$. By the description of $N$, it follows that $N$ is indeed a $k$-lpa.

\s

(2) Secondly, we wish to show the opposite inclusion, namely, $\klpa{k}^R\circ \rtnfamv_t \subseteq \klpa{(k+1)}$. Let $M$ denote a $k$-lpa and let $f$ be a function in $\rtnfamv_t$ witnessed by a certain 1nfa, say, $M_f$.
We write $M^R$ to indicate that $M$ takes the reverse $y^R$ of any input $y$. We define $L_{f,M^R}=\{x\in\Sigma^*\mid \sum_{y\in f(x)}\prob_{M^R}[\text{ $M^R$ accepts $y$ }] > \frac{1}{2}|f(x)|\}$. Our goal is to verify that $L_{f,M^R}\in\klpa{(k+1)}$.
Consider the new machine $N$ that works as follows.
On input $x$ given to a tape, run $M_f$ on $x$  by replacing $x$ with $M_f$'s output strings $y$ produced along computation paths. Once $\cent{y}\dollar$ is written on the tape, run $M^R$ on $(\cent y\dollar)^R$ by moving a tape head leftwards starting at $\dollar$.  Since $M$ is $k$-limited, $N$ must be $(k+1)$-limited. We thus conclude that $L_{f,M^R}\in\klpa{(k+1)}$.
\end{proof}


Finally, the proof of Theorem \ref{LPFA-and-PCFL} is given with the help of Lemmas \ref{PCFL-MV-composition} and \ref{reduction-from-k+1-to-k}.

\begin{proofof}{Theorem \ref{LPFA-and-PCFL}}
By taking ``reversal'' of Lemma \ref{reduction-from-k+1-to-k},
we obtain $\klpa{(k+1)}^R = \klpa{k} \circ \rtnfamv_t^R$. From Lemma  \ref{reduction-from-k+1-to-k} again, it also follows that $\klpa{(k+2)} = \klpa{(k+1)}^R\circ \rtnfamv_t$. Combining these two equalities, we conclude that (*) $\klpa{(k+2)} \subseteq (\klpa{k}\circ \rtnfamv_t^R) \circ \rtnfamv_t$.  We then claim the following inclusion, which is not ``trivial'' due to our nonstandard use of ``$\circ$''.

\begin{claim}\label{claim:LPA-vs-rtNFAMV}
$(\klpa{k}\circ \rtnfamv_t^R) \circ \rtnfamv_t \subseteq \klpa{k}\circ ( \rtnfamv_t^R \circ \rtnfamv_t)$.
\end{claim}

\begin{proof}
Take a function $g\in\rtnfamv_t$ and a probabilistic Turing machine $N$ such that the language $L_{g,N}$ belongs to $(\klpa{k}\circ \rtnfamv_t^R) \circ \rtnfamv_t$. This machine $N$  is further assumed to characterize a language $L_{f^R,M}$ as $L(N) = L_{f^R,M}$
by a certain function  $f\in\rtnfamv_t$ and a certain $k$-lpa $M$.
Associated with this $f$, we take a 1nfa $M_{f^R}$ with a write-once output tape that computes $f^R$. Since $N$ recognizes $L_{f^R,M}$, we focus our attention on the machine $N$ that firstly simulates $M_{f^R}$ to obtain its outcomes and then runs $M$ on the outcomes.

For readability, we write $L'$ for $L_{g,N}$. Next, we intend to verify that this language $L'$ is in fact in $\klpa{k}\circ ( \rtnfamv_t^R \circ \rtnfamv_t)$.
For this purpose, we need to define a function $h\in\rtnfamv_t^R \circ \rtnfamv_t$ and another $k$-lpa $K$ that force $L_{h,K}$ to match  $L'$.

Let us consider the following 1nfa $M'$. Given an input $y^R$ of length $n$, while $M'$ reads $y^R$, $M'$ guesses a length-$n$ string $z=z_1z_2\cdots z_{n-1}z_n$ symbol by symbol and produces $z_1z_2\cdots z_{n-1}$ onto its output tape. Simultaneously, $M'$ runs $M_{f^R}$ on $y$ to obtain its outcomes, say, $z'$. If $z=z'$ and $M_{f^R}$ enters an accepting state, then we set $a=acc$; otherwise, we set $a=rej$. Finally, $M'$ writes down $\track{z_n}{a}$ on the output tape and enters an accepting state.
Note that the outcomes of $M'$ are the strings of the form  $z_1z_2\cdots z_{n_1}\track{z_n}{a}$.
Now, we define $f_{+}^R(y)$ to be the function computed by $M'$.
The desired function $h$ is then set as $h(x) = \{\track{y}{z}\mid y\in g(x), z\in f_{+}^R(y)\}$ for any $x$.
Next, we define the desired $k$-lpa $K$. Given an input $u$ of the form $\track{y}{z}$ with $z=z_1z_2\cdots z_{n_1}\track{z_n}{a}$, if $a=rej$, then we enter both an accepting state and a rejecting state with equiprobability. Otherwise, we simulates $M$ on $z_1z_2\cdots z_n$ using only the lower track of the tape (and the upper track is intact).

Since $M'$ guesses all length-$n$ strings $z$ and never invalidates their outcomes $z'$, we obtain  $| f_{+}^R(y)|=| f_{+}^R(y')|$ for any $y,y'\in g(x)$.
It also follows that, for any $y\in g(x)$, since $|f^R(y)| = |f^R_{+}(y)|$,  $\prob_N[\text{ $N$ accepts $y$ }]$ equals $\sum_{z\in f^R(y)} \prob_M[\text{ $K$ accepts $\track{y}{z}$ }]\cdot \frac{1}{|f_{+}^R(y)|}$.
Since $|h(x)|= |g(x)|| f_{+}^R(y)|$ for any $y\in g(x)$, we conclude that $\sum_{y\in g(x)} \prob_N[\text{ $N$ accepts $y$ }]\cdot \frac{1}{|g(x)|}$ coincides with $\sum_{y\in g(x)}\sum_{z\in f^R(y)} \prob_K[ \text{ $K$ accepts $\track{y}{z}$ }]\cdot \frac{1}{|g(x)|| f_{+}^R(y)|}$, which is equal to $\sum_{\track{y}{z}\in h(x)} \prob_K[\text{ $K$ accepts $\track{y}{z}$ }]\cdot \frac{1}{|h(x)|}$. This implies that $L'=L_{h,K}$, as requested.
\end{proof}

By applying Claim \ref{claim:LPA-vs-rtNFAMV} to (*), we obtain $\klpa{(k+2)} \subseteq \klpa{k}\circ (\rtnfamv_t^R \circ \rtnfamv_t)$.
Lemma \ref{two-rtNFAMV-collapse} then leads to $\klpa{k}\circ (\rtnfamv_t^R \circ \rtnfamv_t) \subseteq  \klpa{k}\circ \rtnfamv_t$.
We conclude that $\klpa{(k+2)} \subseteq \klpa{k}\circ \rtnfamv_t$.  Lemma \ref{PCFL-MV-composition} then implies that $\klpa{(k+2)} \subseteq \klpa{k}$. Since $k$ is arbitrary with $k\geq2$, we conclude that $\klpa{\omega} = \klpa{2}$. A similar argument also works to prove that $\klra{\omega} = \klra{2}$ and $\klua{\omega}=\klua{2}$.
\end{proofof}

\section{Closure Properties of Limited Automata}\label{sec:closure_properties}

The next goal of this paper is to study fundamental closure properties of  language families induced by various types of limited automata  formally defined in Section \ref{sec:formal-definition}.
Given a $k$-ary operation $op$ over languages, a language family $\CC$ is said to be \emph{closed under $op$} if, for any $k$ languages $L_1,L_2,\ldots,L_k$ in $\CC$, the language $op(L_1,L_2,\ldots,L_k)$ always belongs to $\CC$. We will explore basic closure properties of $\klra{k}$, $\klbpa{k}$, and $\klda{k}$. Those properties help us show separations of the language families in the end.

\subsection{Case of Deterministic Limited Automata}\label{sec:closure-deterministic}

It is well-known that $\dcfl$ is not closed under the following operations: finite union, finite intersection, reversal, concatenation, $\lambda$-free homomorphism, and Kleene star (see, e.g., \cite{HU79,WW86}). Since $\twolda=\dcfl$ \cite{PP15}, the same non-closure properties hold for $\twolda$.We wish to further expand $\klda{2}$ to a more general $\klda{k}$ despite the fact that $\klda{k}\neq \dcfl$ for all $k\geq3$ \cite{Hib67}.

\begin{proposition}\label{closure-kLDFA}
For any index $k\geq2$, $\klda{k}$ is closed under none of the following operations: finite union, finite intersection, reversal, concatenation, $\lambda$-free homomorphism, and Kleene star.
\end{proposition}

\begin{proof}
Since $\twolda =\dcfl$ \cite{PP15} and the well-known non-closure properties of $\dcfl$, the proposition is true for the case of $k=2$. It thus suffices to discuss the case of $k\geq3$. Recall from Section \ref{sec:example-language} the special language $L_k$ and strings $w_1\in\{a,b\}$ and $w_i=a^{n_i}b^{m_i}c^{p_i}$ for arbitrary numbers $n_i,m_i,p_i\geq0$.

[finite union]
Given the language $L_k$, for each symbol $x\in\{a,b\}$, we set  $L_k^{(x)}$ to be $\{w\# x \in L_{k} \mid w=w_2\# w_4\#\cdots \# w_{k-1}\# \cdots \#w_5\# w_3\}$.
Since $x$ is fixed, it is not difficult to show that $L_k^{(a)},L_k^{(b)}\in\klda{k}$.
Since $L_{k}$ equals $L_k^{(a)}\cup L_k^{(b)}$, it follows that $L_{k}^{(a)}\cup L_k^{(b)} \notin \klda{k}$. We thus conclude that $\klda{k}$ is not closed under union.

[finite intersection]
This comes from the case of ``finite union'' together with the fact that $\klda{k}$ is closed under complementation.

[reversal]
Let us consider the reversal $L_k^{R}$ of the language $L_k$. If
$w$ denotes $w_2\#w_4\# \cdots \#w_{k-1}\#\cdots \#w_5\#w_3 \# w_1$, then $w^R$ is $w_1^R \# w_3^R \# \cdots \# w_{k-1}^R \# \cdots \# w_4^R\# w_2^R$.
Let us consider the following machine $N$ on such an input $w^R$.
Firstly, $N$ checks if $w_1=a$ or $w_1=b$. If $w_1=a$, then $N$ moves its tape head to the right into the $w_2^R$-region to check the correctness of $w_2^R$. Continue checking $w_2^R, w_3^R,\cdots, w_{k-1}^R$ sequentially to determine whether $w^R\in L_k^R$. The case of $w_1=b$ can be similarly handled.
Since we can implement $N$ as a $k$-lda, $L_k^R$ belongs to $\klda{k}$. Because of $L_k\notin\klda{k}$, $\klda{k}$ is not closed under reversal.

[concatenation]
The following argument is in essence similar to \cite{KM12}.
Assume toward a contradiction that $\klda{k}$ is closed under concatenation.
Recall $L_k^{(a)}$ and $L_k^{(b)}$ for fixed symbols $a$ and $b$, defined above.
Here, we further set $L'_k = \# L_k$ ($=\#(L_k^{(a)}\cup L_k^{(b)})$). From $L_k\notin \klda{k}$, it follows that $L'_k \notin \klda{k}$.
We further define $L_k^{(ab)} =\#L_k^{(a)}\cup L_k^{(b)}$ and $L_0$ to be the set of strings $\#w_2\#w_4\#\cdots \#w_{k-1}\#\cdots \#w_3\#w_1$, where $w_1\in\{a,b\}$ and $w_i=a^{n_i}b^{m_i}c^{p_i}$ for certain integers $n_i,m_i,p_i\geq0$.
The first symbol of any string of $L_k^{(ab)}$ helps us distinguish between $\# L_k^{(a)}$ and $L_k^{(a)}$, and thus $L_k^{(ab)}$ falls to $\klda{k}$. Furthermore, the language $\#^* L_k^{(a)}$ is in $\klda{k}$. It is also obvious that $L_0\in\reg$.

Note that $L'_k$ equals $\#^*L_k^{(ab)}\cap L_0$. Since $\klda{k}$ is closed under intersection with regular languages by Proposition \ref{intersection_with_regular}, we obtain $L'_k\in\klda{k}$.
This is a clear contradiction since $L'_k$ is not in $\klda{k}$.

[$\lambda$-free homomorphism]
Take the languages $L_k^{(a)}$ and $L_k^{(b)}$ defined above. Let $L''_k=\natural L_k^{(a)}\cup \# L_k^{(b)}$. Since the prefixes $\natural$ and $\#$ respectively correspond to the suffixes $a$ and $b$, it is possible to prove that $L''_k\in \klda{k}$. Consider the  $\lambda$-free homomorphism $h$ defined as $h(\natural)=\#$ and $h(\sigma)=\sigma$ for any $\sigma\in\{a,b,\#\}$. The image $h(L''_k)$ matches the language $\#(L_k^{(a)}\cup L_k^{(b)})$, which coincides with $L'_k$ as discussed above. Since $L'_k$ is not in $\klda{k}$,  $\klda{k}$ cannot be closed under $\lambda$-free homomorphism.

[Kleene star]
Assume that $\klda{k}$ is closed under Kleene star. Recall $L_k^{(ab)}$, $L'_k$, and $L_0$ defined above and let $L_{\#}=\{\#\}\cup L_k^{(ab)}$.
Since $L_k^{(ab)}\in\klda{k}$, $L_{\#}$ is also in $\klda{k}$. Consider the Kleene closure $L_{\#}^*$ of $L_{\#}$, which is in $\klda{k}$ by our assumption. Since $L_0\in\reg$, the language $L=(L_{\#})^*\cap L_0$ belongs to $\klda{k}$ because $\klda{k}$ is closed under intersection with regular languages (in Proposition  \ref{intersection_with_regular}). Since $L$ is equal to $\#L_k^{(a)}\cup \#L_k^{(b)}$, which is precisely $L'_k$,  $L'_k$ must be in $\klda{k}$, a contradiction against
$L'_k\notin\klda{k}$. Therefore, $\klda{k}$ is not closed under Kleene star.
\end{proof}

In contrast to Proposition \ref{closure-kLDFA}, there are a few closure properties that $\klda{k}$ can enjoy. Given a language $L$, its Kleene closure $L^*$ is called \emph{marked} if the strings in $L$ end in special symbols that appear only at the end (see, e.g., \cite{HS10}). Such special symbols are called \emph{markers}. For instance, $(L\#)^*$ is a marked Kleene closure for any language $L$ with $\#$ not used in $L$.
We say that a language family $\CC$ is \emph{closed under marked Kleene star} if the marked Kleene closure of any language in $\CC$ belongs to $\CC$. The closure under \emph{marked concatenation} is similarly defined.

\begin{proposition}\label{intersection_with_regular}
Given $k\geq2$, $\klda{k}$ is closed under the following operations: (i) finite intersection and finite union with regular languages and (ii) marked Kleene star and marked concatenation.
\end{proposition}

\begin{proof}
(1) Consider $L=L_1\cap L_2$ for any two languages $L_1\in\klda{k}$ and $L_2\in\reg$. Take a $k$-lda $M_1$ recognizing $L_1$ and a 1dfa $M_2$ for $L_2$.  We then  define a $k$-lpa $N$ for $L$ as follows. On input $x$, start the simulation of $M_1$ on $x$. Whenever $M_1$ visits a new input symbol $\sigma$, $N$ simulates $M_2$ on $\sigma$ without moving its tape head. In this simulation, $N$ keeps both inner states of $M_1$ and $M_2$ internally as its own inner states.
Since $M_1$ is always forced to modify input symbols at the first access, the timing of simulation of $M_2$ is clear to $N$. Thus, $L$ belongs to $\klda{k}$. The closure under finite union with regular languages follows instantly from the equality $\klda{k}=\co\klda{k}$.

(2) Consider a marked Kleene closure $L^*$ for a language $L\in\klda{k}$. We remark that every string in $L^*$ can be parsed uniquely into a series of consecutive strings separated by designated markers.  Using this fact, we can run the following simple algorithm: on each input,  run an underlying $k$-lda by feeding parsed strings one by one until either the $k$-lda enters a rejecting state in the midst of an input or the input is completely processed.
\end{proof}

\subsection{Case of Probabilistic Limited Automata}\label{sec:closure-probabilistic}

We have discussed numerous closure properties of the deterministic model of limited automata in Section \ref{sec:closure-deterministic}.
As a followup to the previous section, let us take a close look at various probabilistic models of limited automata and study their closure properties.
We begin with a simple consequence of Theorem \ref{LPFA-and-PCFL} concerning ``reversal''.

\begin{lemma}
The language families $\klpa{2}$ and $\klra{2}$ are closed under reversal.
\end{lemma}

\begin{proof}
It follows from  Theorem \ref{LPFA-and-PCFL} that $\klpa{3} = \klpa{2}$. Since $\klpa{2}^R\subseteq \klpa{3}$, we immediately obtain $\klpa{2}^R \subseteq \klpa{2}$. The same argument works for $\klra{2}$ as well.
\end{proof}

Recall that $\klda{2}\wedge \klda{2}$ is expressed as $\klda{2}(2)$.

\begin{lemma}\label{finite-union-LRFA}
The language family $\klra{2}$ is closed under finite union but not under finite intersection. For the latter property, more strongly, it follows that $\twolda(2)\nsubseteq \klra{2}$.
\end{lemma}

\begin{proof}
[finite union]
Let $h\in\nat^{+}$ with $h\geq2$. Take $\{L_i\}_{i\in[h]} \subseteq \klra{2}$ and a series $\{M_i\}_{i\in[h]}$ of $k$-lpa's, each $M_i$ of which recognizes $L_i$ with a constant one-side error bound $\varepsilon_i\in[0,1)$.
Consider the following $2$-lpa $N$. On input $x$, choose an index $i\in[h]$ with equiprobability $1/h$ and simulate $M_i$ on $x$. Along each computation path of $M_i$, if $M_i$ enters an accepting state, then $N$ enters the same accepting state; otherwise, $N$ enters a new rejecting state. It is not difficult to see that, if  $x\in \bigcup_{i\in[h]} L_i$, then $N$ accepts $x$ with probability at least $1/h$ multiplied by $1-\varepsilon_i$; otherwise, $N$ rejects $x$ with probability $1$. Hence, $N$ makes one-sided error. We thus conclude that $\bigcup_{i\in[h]}L_i$ belongs to $\klra{2}$.

[finite intersection]
Here, we want to show that $\twolda(2)\nsubseteq \klra{2}$ because this implies the non-closure of $\klra{k}$ under finite intersection.
Consider two languages $L'_1=\{a^nb^nc^m\mid n,m\geq0\}$ and $L'_2=\{a^nb^mc^m\mid n,m\geq0\}$ in $\dcfl$ ($=\klda{2}$).
It is well-known that $L'_1\cap L'_2 =\{a^nb^nc^n\mid n\geq0\}\notin \cfl$. Since $\klra{2}\subseteq \klna{\omega} = \cfl$,
we then obtain $L'_1\cap L'_2\notin\klra{2}$. Since $L'_1\cap L'_2\in\twolda(2)$, it follows that $\klda{k}(2)\nsubseteq \klra{2}$.
\end{proof}


It is not clear that $\klbpa{k}$ is closed under finite intersection and finite union. However, by sharp contrast to Lemma \ref{finite-union-LRFA}, $\klbpa{k}$ is large enough to include $\klda{k}(2)$; more strongly, $\klpa{k}_{\varepsilon}$(2) is included in $\klbpa{k}$ if $\varepsilon<1/6$.

\begin{lemma}\label{union-intersect-LBPA}
Let $k\geq2$ and $\varepsilon\in[0,1/6)$. For any operator $\diamond\in\{\wedge,\vee\}$, $\klpa{k}_{\varepsilon}\diamond \klpa{k}_{\varepsilon} \subseteq \klbpa{k}$ holds; in particular, $\klda{k}(2)\subseteq \klbpa{k}$.
\end{lemma}

\begin{yproof}
It suffices to consider the case of $\diamond=\vee$ because $\klbpa{k}$ is closed under complementation by Lemma \ref{LRFA-not-closed-complement}. Let $\varepsilon\in[0,1/6)$ be any error bound and let $M_1$ and $M_2$ denote two $k$-lpa's working over the same alphabet $\Sigma$ with error probability at most $\varepsilon$. Note that $L=L(M_1)\cup L(M_2)$ is in $\klpa{k}_{\varepsilon}\vee \klpa{k}_{\varepsilon}$.
Let us design a new $k$-lpa $N$ that works as follows. On input $x$, choose an index $i\in\{1,2\}$ uniformly at random, run $M_i$ on $x$.
Along each computation path, if $M_i$ enters an accepting state, $N$ accepts $x$ with probability $1$; otherwise, it accepts $x$ with probability $1/3$ and rejects $x$ with probability $2/3$.
It then follows that $p_{acc,N}(x)$ equals $\frac{1}{2}(p_{acc,M_1}(x)+p_{acc,M_2}(x)) + \frac{1}{6}(p_{rej,M_1}(x)+p_{rej,M_2}(x))$ and $p_{rej,N}(x)$ equals $\frac{1}{3}(p_{rej,M_1}(x)+p_{rej,M_2}(x))$.
Let $\varepsilon'=\varepsilon+\frac{1}{3}$. Note that $\frac{1}{3}\leq \varepsilon'<\frac{1}{2}$ follows from $\varepsilon\in[0,1/6)$.
If $x\in L$, then the error probability of $N$ is at most $2\varepsilon/3 \leq \varepsilon'$ since $p_{rej,M_1}(x)+p_{rej,M_2}(x)\leq 1+\varepsilon$.
In contrast, when $x\notin L$, the error probability of $N$ is at most  $\varepsilon+ \frac{1}{3} \leq \varepsilon'$ since $p_{acc,M_1}(x)+p_{acc,M_2}(x)\leq 2\varepsilon$. Therefore, $L$ is in $\klbpa{k}$.

The second part of the lemma follows immediately from the first part because $\klpa{k}_0 = \klda{k}$ holds.
\end{yproof}

It is, however, unknown that $\klda{k}(d)\subseteq \klbpa{k}$ holds for every index   $d\geq3$.


\begin{lemma}\label{LRFA-not-closed-complement}
For any $k\geq2$, $\klbpa{k}$ is closed under complementation but $\klra{2}$ is not.
\end{lemma}

\begin{yproof}
It is not difficult to show that $\klbpa{k}=\co\klbpa{k}$  for all  indices  $k\geq2$ by swapping between accepting states and rejecting states.
By Lemma \ref{finite-union-LRFA}, $\klra{2}$ is closed under finite union.  If $\klra{2} =\co\klra{2}$, then $\klra{2}$ must be closed under finite intersection. This contradicts the second part of Lemma \ref{finite-union-LRFA}.
\end{yproof}


With the use of closure/non-closure properties of various limited automata discussed so far, we can obtain the following
class separations.

\begin{theorem}\label{separation-all}
For any $k\geq2$, $\klda{k}\subsetneqq \klra{2} \subsetneqq \klbpa{k}$.
\end{theorem}

\begin{proof}
All simple inclusions obviously hold. We need to show the remaining two separations.
Note that $\klda{k}=\co\klda{k}$ for any $k\geq1$. Since $\klra{2}\neq \co\klra{2}$ by Lemma \ref{LRFA-not-closed-complement}, it follows that $\klda{k}\neq \klra{2}$. Similarly, from $\klbpa{k}=\co\klbpa{k}$ by Lemma \ref{LRFA-not-closed-complement}, we obtain $\klra{2}\neq \klbpa{k}$.
\end{proof}

\section{A Brief Discussion and Future Directions}\label{sec:discussion}

Since Hibbard \cite{Hib67} had introduced his model of \emph{scan limited automata} in 1967, the model seemed to have been lost to oblivion for a few decades until Pighizzini and Pisoni \cite{PP14} rediscovered it in 2014.
We have followed their reformulation as \emph{$k$-limited automata} and have examined the behaviors of various models of limited automata, in particular, one-sided error, bounded-error, and unbounded-error probabilistic $k$-limited automata.  We then have proven numerous properties, including closure/non-closure properties as well as collapses and separations of language families induced by these machine models.

For an expected further study on limited automata, we wish to list a few interesting and important open problems, whose solutions would lead to a significant promotion of the better understanding of the limited automata.

\begin{enumerate}\vs{-2}
  \setlength{\topsep}{-2mm}%
  \setlength{\itemsep}{1mm}%
  \setlength{\parskip}{0cm}%

\item We already know  that $\klna{\omega}=\klna{2}$ \cite{Hib67}, $\klua{\omega}=\klua{2}$, $\klra{\omega}=\klra{2}$, and $\klpa{\omega}=\klpa{2}$ (Theorem \ref{LPFA-and-PCFL}).
    This fact means that, e.g.,  $\{\klna{k}\mid k\in\nat^{+}\}$ and $\{\klpa{k}\mid k\in\nat^{+}\}$ do not form a true (infinite)  hierarchy. By contrast, as demonstrated by Hibbard, $\{\klda{k}\mid k\in\nat^{+}\}$ is indeed a true hierarchy within $\cfl$. Does $\{\klbpa{k}\mid k\in\nat^{+}\}$ form a true  hierarchy within $\klpa{2}$? In other words, does $\klbpa{k}\neq \klbpa{(k+1)}$ hold for any $k\geq2$?

\item We have demonstrated  in Lemma \ref{union-intersect-LBPA} that $\klda{k}(2)\subseteq \klbpa{k}$ for any $k\geq2$. Is it possible to expand this inclusion to $\klda{k}(d)\subseteq \klbpa{k}$ for any $d\geq2$? If this is not the case, we can instantly conclude that $\klbpa{k}$ is not closed under intersection.

\item Theorem \ref{LPFA-and-PCFL} implies that $\klbpa{\omega}\subseteq \klpa{\omega}=\klpa{2}$. We are wondering if this inclusion is indeed proper, namely, $\klbpa{\omega}\subsetneqq \klpa{2}$. If $\{\klbpa{k}\mid k\in\nat^{+}\}$ is proven to be a true hierarchy, then $\klpa{2}$ will turn out to have a richer structure than what we currently know.

\item The blank-skipping property has been used in Section \ref{sec:application-BSP} to make a bridge between pushdown automata and 2-limited automata. In particular, we have shown in Corollary \ref{RCFL-BPCFL-PCFL} that $\bpcfl=\mathrm{bs}\mbox{-}\klbpa{2}$ and $\pcfl = \mathrm{bs}\mbox{-}\klpa{2}$. Unfortunately, unlike $\klda{2}$ and $\klua{2}$ (as in Proposition \ref{LPFA-blank-skipping}),
    we do not know that $\klbpa{2}$ and $\klpa{2}$ both satisfy the blank skipping property.
    Is it true that $\klbpa{k}= \mathrm{bs}\mbox{-}\klbpa{k}$ as well as $\klpa{2} = \mathrm{bs}\mbox{-}\klpa{2}$? If this is the case, then we instantly obtain $\klbpa{2}=\bpcfl$ and $\klpa{2}= \pcfl$.

\item It is well-known that $\dcfl$ is not closed under intersection, that is, $\dcfl\neq \dcfl(2)$. More strongly, $\{\dcfl(d)\mid d\in\nat^{+}\}$ forms an infinite hierarchy over $\dcfl$ \cite{LW73} (reproven in \cite{Yam20}). Since $\klda{k}$ extends $\dcfl$, we wonder if,  for each fixed $k\geq2$,  $\{\klda{k}(d)\mid d\in\nat^{+}\}$ also forms an infinite hierarchy over $\klda{k}$. Prove or disprove this statement.

\item The definition of limited automata given in Section \ref{sec:formal-definition} uses two special endmarkers, $\cent$ and $\dollar$. As shown in \cite{Yam21}, these endmarkers can be removed from many variants of pushdown automata without compromising the computational power. It is possible to introduce the model of  ``no-endmarker'' limited automata by eliminating the endmarkers from the original definition of the limited automata. Are those no-endmarker limited automata equivalent in recognition power to the original limited automata with the endmarkers?

\item A \emph{Las Vegas pushdown automaton} is not allowed to err but is permitted to enter a special inner state of ``I don't know'' beside accepting and rejecting states. Such Las Vegas pushdown automata were extensively discussed in \cite{HS10}. A  similar motivation can drive us to explore the basic properties of Las Vegas  probabilistic models of limited automata.

\item In the past literature, much attention has been paid to \emph{unary languages}, which are languages over single-letter alphabets, because of their distinctly unique roles in automata theory. It is known that unary 1npda's are equivalent in power to unary 1dfa's \cite{GR62}. Ka\c{n}eps, Geidmanis, and Freivalds \cite{KGF97} argued that  unary 1ppda's with bounded-error probability are no more powerful than 1dfa's. Can we extend these results to unary $k$-lpa's with bounded-error probability?

\item Hromkovi\v{c} and Schnitger \cite{HS10} discussed, using our notation, two practical language families: $\bigcap_{\varepsilon\in(0,1/2)} \oneppda_{\varepsilon}$ and $\bigcap_{\varepsilon\in(0,1)} \onerpda_{\varepsilon}$. Similarly, we can define $\klbpa{k}^* = \bigcap_{\varepsilon\in(0,1/2)} \klpa{k}_{\varepsilon}$ and $\klra{k}^* = \bigcap_{\varepsilon\in(0,1)}\klra{k}_{\varepsilon}$. It is interesting to study the power and limitation of these language families in comparison to $\klbpa{k}$ and $\klra{2}$.
\end{enumerate}



\let\oldbibliography\thebibliography
\renewcommand{\thebibliography}[1]{%
  \oldbibliography{#1}%
  \setlength{\itemsep}{0pt}%
}
\bibliographystyle{plain}


\end{document}